\documentclass{article}
\pdfoutput=1
\usepackage{iclr2027_conference,times}

\usepackage{amsmath,amsfonts,bm}

\def\eqref#1{equation~\ref{#1}}

\def\1{\bm{1}}

\def\eps{{\epsilon}}

\def\rvs{{\mathbf{s}}}

\def\rvx{{\mathbf{x}}}
\def\rvy{{\mathbf{y}}}

\DeclareMathAlphabet{\mathsfit}{\encodingdefault}{\sfdefault}{m}{sl}
\SetMathAlphabet{\mathsfit}{bold}{\encodingdefault}{\sfdefault}{bx}{n}

\def\gD{{\mathcal{D}}}

\def\gI{{\mathcal{I}}}

\newcommand{\E}{\mathbb{E}}

\makeatletter
\DeclareRobustCommand{\eqref}[1]{\textup{\tagform@{\ref{#1}}}}
\makeatother
\usepackage{algpseudocode}
\usepackage{hyperref}
\usepackage{url}
\usepackage{subcaption}
\usepackage{tikz}
\usetikzlibrary{patterns}
\usepackage{pgfplots}
\pgfplotsset{compat=1.18}

\usepackage{amsthm}
\usepackage{wasysym}
\usepackage{thm-restate}
\usepackage{dsfont}
\usepackage{xspace}
\usepackage{mathtools}
\usepackage[nameinlink,capitalise]{cleveref}

\newtheorem{theorem}{Theorem}
\newtheorem{lemma}[theorem]{Lemma}

\newtheorem{claim}[theorem]{Claim}
\newtheorem{definition}[theorem]{Definition}
\newtheorem{remark}[theorem]{Remark}

\newcommand{\Opt}{\textsc{Opt}\xspace}
\newcommand{\Det}{\textsc{Det}\xspace}
\newcommand{\Rand}{\textsc{Rand}\xspace}

\newcommand{\smart}{\textsc{SmartCow}\xspace}
\newcommand{\Bid}{\textsc{Bid}\xspace}
\newcommand{\define}{:=}
\renewcommand{\eps}{\varepsilon}
\newcommand{\reals}{\mathbb{R}}
\newcommand{\ints}{\mathbb{Z}}
\newcommand{\ratio}{R}
\DeclareMathOperator{\supp}{\mathsf{supp}}
\DeclarePairedDelimiter{\set}{\{}{\}}

\DeclarePairedDelimiter{\vct}{\langle}{\rangle}

\title{Pure Tail Constraints for Online Problems}

\author{Mateusz Basiak, Marcin Bienkowski, Yongho Shin \& Agnieszka Tatarczuk \\
Institute of Computer Science\\
University of Wrocław\\
Wrocław, Poland \\
\texttt{\{mateusz.basiak, marcin.bienkowski, yongho,} \\
\texttt{agnieszka.tatarczuk\}@cs.uni.wroc.pl}
}

\iclrfinalcopy

\begin{document}

\DeclareFontShape{OT1}{ptm}{m}{scit}{<->ssub*ptm/m/sc}{}

\maketitle
\lhead{}
\renewcommand{\headrulewidth}{0pt}

\begin{abstract}
Controlling tail risk is an important objective in online optimization, and recently it has been studied in the context of competitive analysis. Continuing this line of research, we investigate \emph{pure tail constraints}, which capture the tradeoff between expected and worst-case competitiveness. For two fundamental search problems, online bidding and line search, we derive the Pareto-optimal frontiers of this tradeoff. We then investigate another classic problem, TCP acknowledgment, which has structure similar to the iterated ski rental problem. There, we construct an~algorithm whose tradeoff coincides with the known Pareto-optimal tradeoff for ski rental. The lower bounds for this problem are substantially more involved as the problem exhibits adaptive structure: an online algorithm observes requests of the adversary (packet arrivals) and may \emph{adaptively} adjust its actions (acknowledgments) on this basis. We emphasize that all previous work on tail risk in the context of competitive analysis was restricted to non-adaptive problems, where the feedback given to an algorithm was essentially limited to a~binary indicator of whether the algorithm has succeeded or not. Nonetheless, we identify a set of constraints implied by tail bounds in this adaptive setting, and show that they imply a nontrivial lower bound on the TCP acknowledgment problem.

\end{abstract}

\section{Introduction}

\emph{Online optimization} studies sequential decision-making with uncertainty about the future and under adversarial inputs. In many cases, randomization can substantially improve expected performance guarantees of online algorithms. However, it also introduces variability in solution quality and, consequently, the possibility of poor outcomes on individual runs. Typically, the choices in online optimization are irrevocable. An online algorithm cannot in general be rerun, and we cannot pick its best realization afterwards. Hence, an algorithm with a favorable expected performance guarantee may still exhibit a significant tail risk, i.e., a significant probability that its performance is far from the expected value. Controlling such risk is an important objective and has been studied in various fields of online optimization, including expert advice model~\citep{even2006risk}, multi-armed bandits~\citep{sani2012risk,galichet2013exploration,agrawal2021optimal,ayyagari2023risk}, Markov decision processes~\citep{howard1972risk,geibel2005risk,fei2020risk,ghosh2025online}, portfolio selection~\citep{uziel2018growth}, and submodular optimization~\citep{soma2023online}.

\emph{Competitive analysis} is a standard measure of performance for online algorithms~\citep{BorElY98}. We call a (randomized) online algorithm \emph{$\rho$-competitive} if, for every input, its expected cost is at most $\rho$ times the optimal cost with hindsight (see \cref{def:dependable} for the formal statement). Recently, risk control has been incorporated into the competitive analysis of randomized online algorithms, taking into account not only the expected performance of an algorithm but also the tail of its cost distribution.

This risk control took various forms. \citet{christianson2024risk} measured the expected ratio on a~fraction of worst outcomes. Their notion, called conditional value-at-risk (CVaR) competitive ratio, was applied to ski rental and one-max search, and has been later extended to the variant of the Bahncard problem with infinite discount duration by \citet{himmelreich2026risk}. A slightly different approach was introduced by \citet{DiILMV24}. They studied the expected competitive ratio of a randomized algorithm under the constraint that this ratio exceeds a prescribed threshold with at most a given probability. With the requirement of this probability being zero, we get so-called \emph{pure tail constraints}, which are the focus of this paper. This framework was later extended to two-slope ski rental by \citet{cui2025controlling}. See \citet{christianson2024risk} for a more detailed comparison of CVaR and tail-constraint-based approaches.

\subsection{Our Results and Technical Overview}

In this paper, we study \emph{pure tail constraints} in the competitive analysis of three fundamental online optimization problems: \emph{online bidding}, \emph{line search}, and \emph{TCP acknowledgment}. Informally, we call an algorithm \emph{$\gamma$-dependable} if, for every input and every realization of its internal randomness, its performance is within a factor $\gamma$ of the offline optimum with hindsight (see \cref{def:dependable} for the formal statement). By definition, $\gamma$ cannot be smaller than the optimal deterministic competitive ratio. Moreover, as $\gamma$ tends to infinity, the constraint becomes vacuous, and we expect to recover the optimal randomized competitive ratio. We investigate the tradeoff between these two extremes. For the classical ski rental problem, \citet{DiILMV24} characterized the Pareto-optimal frontier of this tradeoff. We extend this investigation to three other canonical online problems.

\paragraph{Non-adaptive problems.}

We start with non-adaptive problems where the feedback given to an~algorithm is just an indicator whether the algorithm has succeeded or not. \emph{Online bidding} is a~fundamental online problem that captures the challenge of searching for an~unknown target value \citep{ChrKen06}. An algorithm submits a sequence of bids until one reaches or exceeds the target, where the cost is defined as the sum of all submitted bids, including the final one. The optimal deterministic and randomized competitive ratios for this problem are $4$ and $e$, respectively~\citep{ChKeNY08}. Solutions for online bidding have been used as building blocks in algorithms and lower bounds for a variety of problems ranging from machine scheduling~\citep{HaScSW97,EbeSga09,ELMMMS10} and strip packing~\citep{YeHaZh11}, through multi-level aggregation~\citep{BBBCDF16} and minimum latency tours~\citep{GoeKle98,ChChFM04}, to clustering~\citep{ChKeNY08}. In this paper, we characterize the Pareto-optimal frontier of the tradeoff between dependability and expected competitiveness for online bidding as follows (see \cref{fig:bidding}):
\begin{restatable}{theorem}{thmbid} \label{thm:res_bidding}
    Fix $r \in [2, e]$ and let $\gamma(r) \define \frac{r^2}{r-1}$. There exists a $\gamma(r)$-dependable randomized algorithm for online bidding whose expected competitive ratio is $\frac{r}{\ln r}$. Conversely, every $\gamma(r)$-dependable randomized algorithm for online bidding has expected competitive ratio at least $\frac{r}{\ln r}$.
\end{restatable}

The dependability $\gamma(r)$ is always at least $4$. When $r = e$, we have $\gamma(r) = e^2 / (e - 1) \approx 4.3$, and the expected competitive ratio reaches the unconstrained randomized optimum $e$.

The upper bound is an adaptation of the algorithm of \citet{ChKeNY08}; our main contribution here is the lower bound. We use Yao's min-max principle and analyze a $\gamma(r)$-dependable deterministic algorithm against a random threshold. Its expected ratio is a function of the ratios between consecutive bids, and our main technical tool (\cref{lem:step_ratio}) shows that $\gamma(r)$-dependability constrains these ratios. Roughly speaking, the algorithm cannot do substantially better than the geometric sequence in which every bid is $r$ times larger than the previous one, up to lower-order terms. Details can be found in \cref{sec:bid}.

Next, we extend our results on online bidding to \emph{line search}~\citep{BaCuRa93}. In this classic problem, an algorithm must locate a target at an unknown position on the line while minimizing the total distance traveled before finding it. Line search is closely related to online bidding~\citep{chrobak2006sigact}, and based on this connection, we also provide the Pareto-optimal dependability-competitiveness frontier for this problem as follows (see \cref{fig:line_search} and \cref{sec:line} for details).

\begin{restatable}{theorem}{thmline} \label{thm:res_line}
    Fix $r \in [2, r_0]$, where $r_0 \approx 3.591$ is the solution of $r = \frac{1+r}{\ln r}$, and let $\gamma(r) \define 1 + \frac{2r^2}{r-1}$. There exists a $\gamma(r)$-dependable randomized algorithm for line search whose expected competitive ratio is $1 + \frac{1+r}{\ln r}$. 
    Conversely, every $\gamma(r)$-dependable randomized algorithm for line search has expected competitive ratio at least $1 + \frac{1+r}{\ln r}$.
\end{restatable}

\paragraph{An adaptive problem.}

The final problem we consider in this paper is \emph{TCP acknowledgment}~\citep{DoGoSc01}, in which packets sequentially arrive, and an algorithm must acknowledge each packet after its arrival. Each acknowledgment clears all pending (unacknowledged) packets. The objective is to minimize the sum of the number of acknowledgments sent and the total latency incurred by all packets before they are acknowledged. The optimal deterministic competitive ratio is $2$~\citep{DoGoSc01}. By exploiting a connection to ski rental, \citet{KaKeRa03} subsequently obtained a randomized algorithm with competitive ratio $e/(e-1)$, while the matching lower bound was established by \citet{Seiden00}. The problem is also known as \emph{lot sizing}, and has been extensively studied in the operations research literature; see, e.g., a book by~\citet{Kimms97}.

Before we describe our results, we emphasize a key difference between TCP acknowledgment and the previous two problems. In online bidding and line search, an algorithm is necessarily \emph{non-adaptive}: the feedback available to an algorithm is essentially limited to whether the process has terminated; once the goal is reached, no further decisions are required. In contrast, in TCP acknowledgment, a deterministic algorithm can be \emph{adaptive}: it can observe packets as they arrive and may adjust its future acknowledgments to them. This adaptivity makes the analysis of TCP acknowledgment more challenging than that of online bidding and line search. We emphasize that \emph{all problems studied so far in CVaR and tail-constraint-based frameworks were not adaptive}: ski rental, one-max search, two-slope ski rental, and Bahncard with infinite discount duration. Our work is thus the first to study pure tail constraints for an adaptive online problem.

Inspired by the connection identified by \citet{KaKeRa03}, we construct an algorithm for TCP acknowledgment that attains the same dependability-competitiveness tradeoff as the ski rental algorithm of \citet{DiILMV24} (see \cref{fig:tcp}):
\begin{restatable}{theorem}{tcpub} \label{thm:res_tcp}
    Fix $r \in (0, 1]$ and let $\gamma(r) \define 1 + \frac{1}{r}$. There exists a $\gamma(r)$-dependable randomized algorithm for TCP acknowledgment whose expected competitive ratio is $1 + \frac{e^{r}}{(1+r) \cdot e - e^{r}}$. 
\end{restatable}

\begin{figure}[t]
\centering
\begin{subfigure}[t]{0.31\textwidth}
    \centering
    \includegraphics[width=\textwidth]{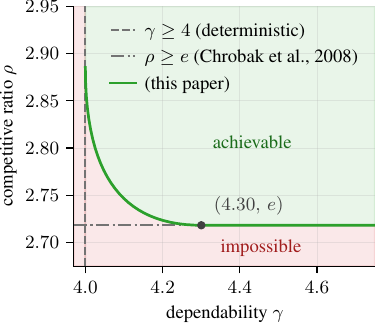}
    \caption{Online bidding, \cref{thm:res_bidding}.}
    \label{fig:bidding}
\end{subfigure}
\hfill
\begin{subfigure}[t]{0.285\textwidth}
    \centering
    \includegraphics[width=\textwidth]{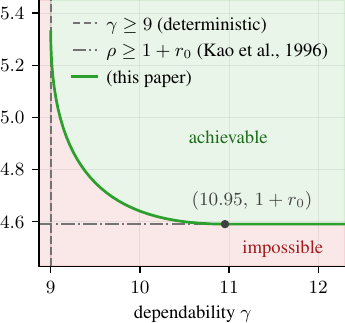}
    \caption{Line search, \cref{thm:res_line}.}
    \label{fig:line_search}
\end{subfigure}
\hfill
\begin{subfigure}[t]{0.36\textwidth}
    \centering
    \includegraphics[width=\textwidth]{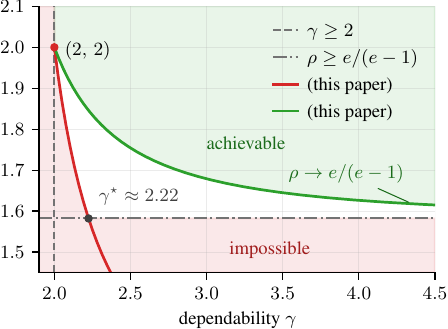}
    \caption{TCP ack.,~\cref{thm:res_tcp,thm:res_tcp_lb}.}
    \label{fig:tcp}
\end{subfigure}
\caption{The dependability-competitiveness tradeoffs for online bidding, line search, and TCP acknowledgment. The curves separate the achievable pairs from the unachievable ones. For TCP acknowledgment, the two bounds do not coincide: every pair above the upper curve is achieved by the algorithm of \cref{thm:res_tcp}, and no pair below the lower curve is achievable by \cref{thm:res_tcp_lb}. The two curves meet at the point $(2,2)$, and our lower bound improves on the randomized one of $e/(e-1)$ for $\gamma < \gamma^\star \approx 2.22$.}
\label{fig:bidding_line}
\end{figure}

Adaptivity makes a matching lower bound for TCP acknowledgment much harder to obtain. An~algorithm plays multiple rounds against the adversary, and it can adapt its future decisions to the packets observed so far. Nevertheless, we establish a nontrivial lower bound on the achievable tradeoff (see also \cref{fig:tcp}):

\begin{restatable}{theorem}{tcplb} \label{thm:res_tcp_lb}
    Fix $r \in [\frac{\sqrt5-1}{2}, 1]$ and let $\gamma(r) \define 1 + \frac{1}{r}$. Every $\gamma(r)$-dependable randomized algorithm for TCP acknowledgment has expected competitive ratio at least
    $\rho(r) = 1 + \frac{(2+r) \cdot r^{2+r}}{(1+r)^2 - r^{2+r}}$.
\end{restatable}

To prove this bound, we extend the construction of \citet{Seiden00}. The adversary issues bursts of packets, each much larger than all the preceding ones together, and decides on the next burst only after seeing the algorithm acknowledge the previous one. We show that $\gamma(r)$-dependability forces structural properties on every deterministic algorithm. For example, it cannot acknowledge the first burst too early. This restricts the support of the randomized algorithms we need to consider. We then construct a probability distribution over instances with at most two bursts and apply Yao's min-max principle to this restricted class.

\section{Preliminaries}\label{sec:preliminaries}

All problems which we study are online cost minimization problems: an instance $\sigma$ is revealed to an~algorithm over time, and the algorithm must react to the revealed parts without any knowledge about the future. We write $\Opt(\sigma)$ for the optimal offline cost, $\Det(\sigma)$ for the cost of a deterministic online algorithm \Det on instance $\sigma$, and
\[
    \ratio(\Det, \sigma) \define \frac{\Det(\sigma)}{\Opt(\sigma)}
\]
for the solution quality of \Det on $\sigma$.

\begin{definition}
    \label{def:dependable}
    A randomized algorithm is a probability distribution over deterministic algorithms; we denote such an algorithm by \Rand and its support by $\supp(\Rand)$. It is \emph{$\rho$-competitive}\footnote{This is so-called \emph{strict competitiveness}; for a non-strict one, see \cref{sec:competitive_ratio}.} if $\E_{\Det \sim \Rand}[\ratio(\Det, \sigma)] \leq \rho$ for every instance~$\sigma$, and \emph{$\gamma$-dependable} if $\ratio(\Det, \sigma) \leq \gamma$ for every $\Det \in \supp(\Rand)$ and every instance $\sigma$. 
\end{definition}

Note that $\gamma$-dependability corresponds to the pure tail constraint $(\gamma, 0)$ of \citet{DiILMV24}: the competitive ratio never exceeds $\gamma$, regardless of random choices of \Rand. 

\paragraph{Naming convention.}

Note that for a deterministic algorithm, the notions of competitiveness and dependability coincide. To avoid confusion, \emph{we use dependability whenever we deal with deterministic algorithms}, and to emphasize the distinction we use the term \emph{expected competitiveness} when we deal with the competitiveness of randomized algorithms. Throughout the paper, \Det stands for a deterministic algorithm, \Rand for a randomized one, $\sigma$ for an instance, and $\pi$ for a probability distribution over instances. A parameter on which an algorithm depends is written as a subscript.

\paragraph{Min-max principle.}

In our lower bounds, we show that, for fixed $\gamma$ and $\rho$, no $\gamma$-dependable randomized algorithm has expected competitive ratio better than $\rho$. To this end, we use a slightly modified version of Yao's min-max principle~\citep{Yao77}, stated below. Its proof is given in \cref{sec:competitive_ratio} for the sake of completeness.

\begin{restatable}[Min-max principle]{lemma}{minmax} \label{lem:minmax}
    Let $\gD^\star \subseteq \gD$ be two sets of deterministic algorithms and let $\pi$ be a~probability distribution over instances. Assume that $\E_{\sigma \sim \pi}[\ratio(\Det, \sigma)] \geq \rho$  for every deterministic algorithm $\Det \in \gD$. Then, every randomized algorithm $\Rand$ with $\supp(\Rand) \subseteq \gD^\star$ has expected competitive ratio at least $\rho$.
\end{restatable}

In most applications of this lemma, $\gD^\star$ is the set of all $\gamma$-dependable deterministic algorithms and $\gD$ is the set of all deterministic algorithms satisfying a certain property implied by $\gamma$-dependability. The randomized algorithms in the lemma are then exactly the $\gamma$-dependable randomized algorithms.

\subsection{Problem Definitions}
\label{sec:problems}

All three problems below are classic online problems. The first two problems (online bidding and line search) have a single unknown parameter that describes both the instance and the optimal solution, and are in this sense similar to the ski rental problem. The third problem (TCP acknowledgment), while having a rent-or-buy structure similar to that of the ski rental problem, has more complex instances.

\paragraph{Online bidding.}

In \emph{online bidding} (see, e.g., \citealp{ChKeNY08}), an adversary fixes an unknown threshold $T \geq 1$.\footnote{In the literature, the condition $T \geq 1$ is sometimes dropped. In such a case, one either allows non-strict competitiveness or defines the bids as a bi-infinite sequence that may start with an infinitesimally small value.} The algorithm submits an increasing sequence of bids $x_1 < x_2 < \ldots$; a~bid is \emph{successful} if it is at least $T$. The process stops at the first successful bid, and the algorithm pays the sum of all submitted bids. The optimal offline cost is $T$, achieved by the single bid $T$, so an algorithm whose first successful bid is $x_k$ has solution quality $(\sum_{i \leq k} x_i) / T$.

\paragraph{Line search.}

In \emph{line search}, also known as the \emph{cow-path problem} (see, e.g., \citealp{BaCuRa93}), a searcher (a cow called Bessie) starts at the meeting point of two semi-infinite paths. It must find a target (a gate to a grazing field) placed at an unknown distance $T \geq 1$ from the start, on one of the two paths. The searcher learns the location of the target only upon reaching it, and pays the total distance traveled. The optimal offline cost is $T$, achieved by walking directly to the target.

\paragraph{TCP acknowledgment.}

In \emph{TCP acknowledgment} \citep{DoGoSc01}, an instance is a finite set of packets; packet $j$ arrives at time $a_j \geq 0$ and the algorithm learns about it only then. At any moment, the algorithm may send an acknowledgment. This costs $1$ and clears all pending packets. The cost of a solution is the number of acknowledgments plus the total latency $\sum_j (\text{acknowledgment time of packet $j$} - a_j)$, so the algorithm has to trade acknowledgments against latency.

Related work on these and similar problems is discussed in \cref{sec:related}. 

\section{Online Bidding}
\label{sec:bid}

In this section, we discuss \cref{thm:res_bidding}. Its proof along with the proofs of the lemmas stated below is given in \cref{sec:bid_proofs}; here we focus on the main ideas.

We fix $r \in [2, e]$ and let $\gamma(r) \define r^2/(r-1)$. While our main contribution in this part is the lower bound, we start with a conceptually simpler upper bound. It is attained by the parametrized version of the optimal randomized algorithm by \citet{ChKeNY08}, defined below.

\begin{quote}
\emph{Algorithm $\Bid_r$: first, choose a value $\xi \in [0, 1)$ uniformly at random. Next, submit the increasing bids $r^{\xi}, r^{\xi+1}, r^{\xi+2}, \ldots$ until the first successful one.}
\end{quote}

\begin{restatable}{lemma}{bidub} \label{lem:bid_ub}
    Randomized algorithm $\Bid_r$ is $\gamma(r)$-dependable and its expected competitive ratio is at most~$r/\ln r$.
\end{restatable}

In the proof of the above lemma, we argue that the solution quality of $\Bid_r$ is at most $r^{\zeta+1}/(r-1)$, where $\zeta \in [0,1)$ is the fractional part of $\log_r$ of the last bid, which is distributed uniformly. This gives us both the expected competitive ratio and the dependability, depending on whether we take the expectation over $\zeta$ or its worst case.

For the lower bound, we use the min-max principle (\cref{lem:minmax}) and construct an appropriate probability distribution over instances.

\begin{restatable}{lemma}{detbid} \label{lem:det_bid}
    Fix $\delta > 0$. There exists a probability distribution $\pi(\delta)$ over the inputs, such that every $\gamma(r)$-dependable algorithm $\Det$ satisfies $\E_{\sigma \sim \pi(\delta)}[\ratio(\Det, \sigma)] \geq \frac{r}{\ln r} - \delta$.
\end{restatable}

The constructed probability distribution $\pi(\delta)$ is supported on thresholds $T$ from the range $[1, U]$ for a sufficiently large $U$, and we write $v_t$ for its density at $t$. Any $\gamma(r)$-dependable deterministic algorithm $\Det$ for such instances can be fully described by an increasing sequence of bids $x_1 < x_2 < \dots$, all from range~$[1, U]$; we also set $x_0 \define 1$. $\Det$ places bid $x_i$ if all its previous bids were unsuccessful, which happens with probability $\int_{x_{i-1}}^U v_t \, dt$, and the expected contribution of this bid to the solution quality is $x_i \cdot \int_{x_{i-1}}^U (v_t/t) \, dt$. To simplify the latter term, we choose $v_t = 1/(t\cdot \ln U)$, which allows us to eventually obtain 
\begin{equation} \label{eq:bid_cost}
    \E_{\sigma \sim \pi}[\ratio(\Det, \sigma)] = \frac{1}{\ln U} \cdot \sum_{i=1}^n \frac{x_i}{x_{i-1}} - O\left(\frac{1}{\ln U}\right).
\end{equation}

To lower-bound \eqref{eq:bid_cost}, we observe that after the $k$-th bid, the total cost of $\Det$ is $\sum_{i=1}^k x_i$, while the adversary can terminate the instance choosing a threshold $T$ arbitrarily close to $x_{k-1}$. Since the algorithm is $\gamma(r)$-dependable, this imposes the following constraints:
\[
    \sum_{i=1}^k  x_i \leq \gamma(r) \cdot x_{k-1} = \frac{r^2}{r-1} \cdot x_{k-1} \qquad \text{for every $k \in \set{1, \dots, n}$}.
\]
The technical core of this part is the \emph{step ratio lemma} stated below, which we use to provide a lower bound on \eqref{eq:bid_cost}. It shows that, up to lower-order terms, \eqref{eq:bid_cost} is minimized when all consecutive ratios $x_i / x_{i-1}$ equal~$r$. We prove the lemma in \cref{sec:step_ratio_proof}. Since we reuse it for line search, we state it in a generalized form; for online bidding, we set $c = 0$.

\begin{restatable}[Step ratios]{lemma}{stepratio}
    \label{lem:step_ratio}
    Let $1 = x_0, x_1, \dots, x_n$ be a sequence of positive reals. Fix $c \geq 0$ and let $\rho_c$ denote the unique solution of $\rho \cdot \ln \rho = c + \rho$. Fix a real $r \in [2,\rho_c]$ and $\eps \in (0, \frac{c+r}{\ln r})$. Let $S_k \define \sum_{i=1}^{k} x_i$, and assume that
    \[
        S_k \leq \frac{r^2}{r-1} \cdot x_{k-1} \qquad \text{for every $k \in \set{1, \dots, n}$}.
    \]
     There is a constant $F$, depending on $c$, $r$ and $\eps$ only, such that 
    \[
        c \cdot n + \sum_{i=1}^{n} \frac{x_i}{x_{i-1}} \geq \left( \frac{c+r}{\ln r} - \eps \right) \cdot \ln x_n - F .
    \]
\end{restatable}

\section{Line Search} \label{sec:line}

In this section we focus on presenting ideas for proving \cref{thm:res_line}; the complete proof is deferred to \cref{sec:line_proofs}. Let $r_0$ be defined as in the theorem, i.e., as the solution to $r = (1+r) / \ln r$. Throughout this section, we fix $r \in [2, r_0]$ and let $\gamma(r) \define 1 + 2 \cdot r^2/(r-1)$.

The upper bound of \cref{thm:res_line} is attained by the algorithm of \citet{KaReTa96} called \smart, parametrized by a number $r \in (1, +\infty)$. \citet{KaReTa96} show that $\smart_r$ is $(1 + (1 + r) /\ln r)$-competitive. We complement this result by showing that it is also $\gamma(r)$-dependable for $r \in [2, r_0]$, and thus achieves the tradeoff of \cref{thm:res_line}. The proof is deferred to \cref{sec:line_proofs}.

We obtain the lower bound from the min-max principle of \cref{lem:minmax} and the lemma below. Its proof is deferred to \cref{sec:line_proofs}, and we provide the main ideas here.

\begin{restatable}{lemma}{detline} \label{lem:det_line}
    Fix $\delta > 0$. There exists a probability distribution $\pi(\delta)$ over the inputs, such that every $\gamma(r)$-dependable algorithm $\Det$ satisfies $\E_{\sigma \sim \pi(\delta)}[\ratio(\Det, \sigma)] \geq 1 + \frac{1 + r}{\ln r} - \delta$.
\end{restatable}

Similarly to the proof of \cref{lem:det_bid}, probability distribution $\pi(\delta)$ is defined only on values from the interval $[1, U]$ on both paths, and the probability density of the target $T$ at any $t \in [1,U]$ is $v_t = 1/(2t\cdot \ln U)$. Any deterministic $\gamma(r)$-dependable algorithm $\Det$ can be described by 
\begin{itemize}
    \item a sequence of numbers $x_1, x_2, \ldots$ from $[1,U]$ satisfying
        $x_1 < x_3 < x_5 < \ldots$ and $x_2 < x_4 < x_6 < \ldots$, and
    \item a starting direction.
\end{itemize}
We also set $x_0 \define 1$. $\Det$ first moves to points $x_1, x_2, x_3, \ldots$ on alternating paths, starting with $x_1$ in the chosen starting direction, until it finds the target.

If $\Det$ visits $k$ of these points and finds the target $T$ on its way to the $(k+1)$-st point, its total cost is $T + 2 \cdot \sum_{i=1}^k x_i$. The algorithm arrives at $x_i$ if and only if $T$ is on the same path and larger than~$x_i$, or on the other path and larger than $x_{i-1}$; the probability of this event is $\int_{x_{i-1}}^U v_t \, dt + \int_{x_i}^U v_t \, dt$. This implies a bound similar to \eqref{eq:bid_cost}, albeit having extra terms because of $T$ in the cost and the second integral.

Note that the cost of $\Det$ finding the target after visiting $k$ points is $T + 2 \cdot \sum_{i=1}^k x_i$, and the adversary can choose $T$ arbitrarily close to $x_{k-1}$. Since $\Det$ is $\gamma(r)$-dependable, it has to satisfy the following property on its partial cost: 
\[
    x_{k-1} + 2\cdot \sum_{i=1}^k  x_i \leq \gamma(r) \cdot x_{k-1} = \left(1 + \frac{2r^2}{r - 1}\right) \cdot x_{k-1} \qquad \text{for every $k \geq 1$}.
\]
This is equivalent to the assumption required in \cref{lem:step_ratio}, which we may use (this time with $c=1$) to prove \cref{lem:det_line}.
\section{Upper Bound for TCP Acknowledgment}
\label{sec:ub_tcp}

In this section, we show the main ideas behind the proof of \cref{thm:res_tcp}; the formal arguments are deferred to \cref{sec:app_tcpub}. To this end, we define an algorithm $\Rand_r$, parametrized by \mbox{$r \in (0, 1]$}. This algorithm is inspired by that of \citet{KaKeRa03}, where we modify the probability distribution to guarantee the dependability of $\gamma(r) = (1 + 1/r)$.

\subsection{Dependability}

We start by defining a family of deterministic algorithms $\{\Det_z\}_{z \in (0, 1]}$ introduced by \citet{KaKeRa03}. Fix $z \in (0, 1]$, and let $T_0 \define 0$ denote the very start of the execution. For $i \in \ints_{> 0}$, $\Det_z$ performs the $i$-th acknowledgment at time $T_i$ inductively defined as the first time for which there exists a time $\tau_i \in [T_{i-1}, T_i]$ such that 
\begin{equation} \label{eq:tcp:detzdef}
    P(T_{i-1}, \tau_i) \cdot (T_i - \tau_i) = z,
\end{equation}
where $P(T, T')$ denotes the number of packets that have arrived in time interval $(T, T']$ for any time points $T \leq T'$.

\begin{figure}
    \centering
\begin{tikzpicture}[scale=1]
    \draw[->] (0,0) -- (10.5,0) node[right] {\small{Time}};
    \draw[->] (0,0) -- (0,3);
    \node[rotate=90, anchor=west] at (-0.3,0.4) {\small{Packet arrivals}};

    \draw[thick, black, dashed] (0,0) -- (10, 3);

    \foreach \a/\b in {1.0/0, 3.0/0.6, 5.0/1.2, 7.0/1.8, 9/2.4} { \draw[draw=blue!50, fill=blue!20,preaction={pattern=north east lines, pattern color=blue!20}] ( \a, \b) rectangle ++(1.0, 0.3);}
      
    \draw[blue, thick] (0,0) -- (2,0) -- (2, 0.6) -- (4, 0.6) -- (4, 1.2) -- (6, 1.2) -- (6, 1.8) -- (8, 1.8) -- (8, 2.4) -- (10, 2.4) -- (10, 3) -- (10.2, 3);

    \draw[red, very thick, densely dotted] (0,-0.02) -- (1.5,-0.02) -- (1.5, 0.45) -- (4.2, 0.5) -- (4.2, 1.26) -- (4.6, 1.26) -- (4.6, 1.38) -- (5.6, 1.38) -- (5.6, 1.68) -- (7, 1.68) -- (7, 2.1) -- (7.7, 2.1) -- (7.7, 2.31) -- (8.6, 2.31) -- (8.6, 2.58) -- (10.1, 2.58);

    \draw[black, dash pattern=on 3pt off 6pt] (1,0) -- (1, 0) node[below] {$\tau_1$};
    \draw[black, dash pattern=on 3pt off 6pt] (2,0) -- (2, 0) node[below] {$T_1$};
    \draw[black, dash pattern=on 3pt off 6pt] (3,0.6) -- (3, 0) node[below] {$\tau_2$};
    \draw[black, dash pattern=on 3pt off 6pt] (4,0.6) -- (4, 0) node[below] {$T_2$};
    \draw[black, dash pattern=on 3pt off 6pt] (5,1.2) -- (5, 0) node[below] {$\tau_3$};
    \draw[black, dash pattern=on 3pt off 6pt] (6,1.2) -- (6, 0) node[below] {$T_3$};
    \draw[black, dash pattern=on 3pt off 6pt] (7,1.8) -- (7, 0) node[below] {$\tau_4$};
    \draw[black, dash pattern=on 3pt off 6pt] (8,1.8) -- (8, 0) node[below] {$T_4$};
    \draw[black, dash pattern=on 3pt off 6pt] (9,2.4) -- (9, 0) node[below] {$\tau_5$};
    \draw[black, dash pattern=on 3pt off 6pt] (10,2.4) -- (10, 0) node[below] {$T_5$};
\end{tikzpicture}
\caption{The latency and the acknowledgments of \Opt (red dotted line) and of $\Det_z$ (blue solid line). The sequence of arrivals $\sigma$ is depicted as the dashed diagonal function. The filled rectangles are the latencies that caused the acknowledgments of $\Det_z$. Each filled rectangle has area $z$.}
\label{fig:tcp_dependable}
\end{figure}
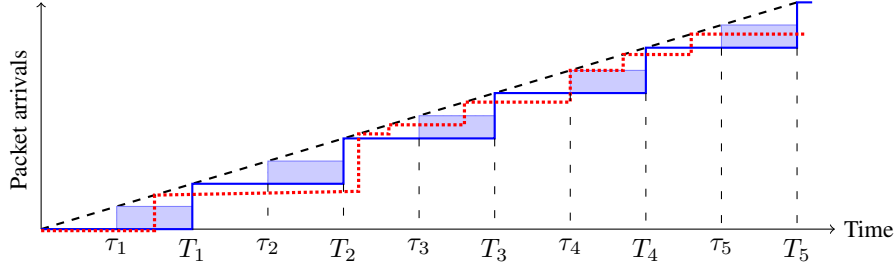

To provide intuitions behind $\Det_z$, \cref{fig:tcp_dependable} shows an illustration of the behaviors of $\Det_z$ and $\Opt$, adapted from \citet{KaKeRa03}. The black dashed line depicts the number of packets that have arrived so far, simplified to a straight line. The blue solid line and the red dotted line represent the numbers of packets acknowledged by $\Det_z$ and $\Opt$, respectively, where every vertical segment of these lines corresponds to an acknowledgment of the respective algorithm.

Note that the $i$-th acknowledgment of $\Det_z$ at time $T_i$ is performed only when the unacknowledged packets until time $\tau_i$ together incur latency $z$ in time interval $[\tau_i, T_i]$, as instructed by the trigger condition \eqref{eq:tcp:detzdef}. This latency cost is represented in \cref{fig:tcp_dependable} by a blue filled rectangle with area exactly~$z$.
Observe that these rectangles are disjoint, and the $i$-th rectangle touches the black dashed line at time $\tau_i$.

We show the dependability of $\Det_z$:

\begin{restatable}{lemma}{dependability} \label{lem:tcp_dependable}
    For every $z \in (0, 1]$, $\Det_z$ is $\gamma(z)$-dependable.
\end{restatable}

The proof is twofold. First, we show that $\Opt$ must incur at least cost $z$ between any two consecutive acknowledgments by $\Det_z$, implying that the total number of acknowledgments by $\Det_z$ is upper-bounded within a factor $1/z$ from the optimal cost. Second, we prove that, for any set of packets acknowledged at the same time by $\Opt$, the \emph{additional} latency cost incurred on them by $\Det_z$ cannot exceed $z$ by the definition of $\Det_z$. Hence, the total latency cost of $\Det_z$ is bounded by the total cost of the optimal solution. These bounds together imply the lemma.

\subsection{Competitive Algorithm}

Since $\Det_z$ is $\gamma(r)$-dependable for any $z \in [r, 1]$, any randomized algorithm supported by $\{\Det_z\}_{z \in [r, 1]}$ is also $\gamma(r)$-dependable. Hence, we aim to attain a probability distribution over $\{\Det_z\}_{z \in [r, 1]}$ minimizing the expected competitive ratio, resulting in the following algorithm.

\begin{quote}
    \emph{Algorithm $\Rand_r$: draw a value $z \in [r, 1]$ from a distribution that has
    \begin{itemize}
        \item a discrete mass of $q_r \define \frac{r \cdot e^r}{(r+1) \cdot e - e^r}$ at $z = r$, and
        \item a density $p(z) \define \frac{r+1}{(r+1) \cdot e-e^r} \cdot e^z$ for any $z \in (r, 1]$,
    \end{itemize}
    and run $\Det_z$.
    }
\end{quote}

\begin{restatable}{lemma}{tcpublem} \label{lem:tcp_ub}
    $\Rand_r$ has expected competitive ratio $1 + \frac{e^{r}}{(1+r) \cdot e - e^{r}}$.
\end{restatable}

Note that, together with \cref{lem:tcp_dependable}, this lemma implies \cref{thm:res_tcp}. To prove this lemma, we exploit the properties derived from \citet{KaKeRa03}.

\section{Lower Bound for TCP Acknowledgment: Proof Overview}
\label{sec:lb_tcp_overview}

In this section, we describe the main ideas behind the proof of \cref{thm:res_tcp_lb}; the full proof is given in \cref{sec:lb_tcp}. Throughout this section, we fix $r \in [(\sqrt5-1)/2, 1]$ and write $\gamma(r) \define 1 + 1/r$ and $c \define r^2 + r - 1 \in [0, 1]$.

We recall that the deterministic algorithm for the problem can be adaptive: it observes packets as they arrive and may adjust its future acknowledgments to them. Such an algorithm is thus a function which maps the choices of the adversary made so far to the delay of the next acknowledgment. Therefore, we analyze the consequences of $\gamma(r)$-dependability by playing \emph{multiple rounds} against the algorithm. In each round, the adversary waits for the algorithm to acknowledge, observes the latency at which this happens, and only then it decides on the next part of the input.

\paragraph{Structural consequences of dependability.}

In the first part of the proof (\cref{sec:latencies}), we study a deterministic algorithm \Det and identify properties which every $\gamma(r)$-dependable algorithm has to satisfy. Following \citet{Seiden00}, we run \Det on instances consisting of \emph{bursts} of packets, each much larger than all the preceding ones together. We call such instances \emph{flat}. The latency of a burst is the total latency accrued by its packets so far. The adversary issues burst $i+1$ when burst~$i$ reaches latency $y_i$, and we write $x_i$ for the latency at which \Det acknowledges burst~$i$. On flat instances, the latency accrued between two bursts is essentially caused by the last burst only, and the solution that acknowledges once, at the last burst, has cost close to $1 + \sum_i y_i$.

Right after \Det acknowledges burst $i$, we define its \emph{post-ack budget} 
$\Phi^+_i$ as $\gamma(r)$ times the cost of the solution that acknowledges once, at this very moment, minus the cost that \Det has paid so far. On flat instances, the former cost is close to $1 + \sum_{j < i} y_j + x_i$. The definition reflects a possible aggressive move of the adversary: it may issue burst $i+1$ right after \Det's acknowledgment of burst~$i$. Then, $\Phi^+_i$ is the amount that $\gamma(r)$-dependability still allows \Det to spend on the future bursts. As the instance may also end at this moment, it holds that $\Phi^+_i \geq 0$.

\paragraph{The penalizing suffix.}

The main ingredient of this part is a threat that the adversary may pose after any acknowledgment of \Det. The adversary may append a \emph{penalizing suffix}, i.e., a sequence of bursts, each issued immediately after \Det acknowledges the previous one. In each round $m$ of this suffix, \Det pays $1 + x_m$ for its acknowledgment, while the cost of the single-acknowledgment solution grows by $x_m$ only. Thus, roughly speaking, the post-ack budgets in consecutive rounds satisfy 
\begin{align*}
    \Phi^+_m 
    & \leq \Phi^+_{m-1} + \gamma(r) \cdot x_m - (1 + x_m) \\
    & = \Phi^+_{m-1} + (\gamma(r) - 1) \cdot x_m - 1 \\
    & \leq \gamma(r) \cdot (\Phi^+_{m-1} - 1),
\end{align*}
where the last inequality uses $x_m \leq \Phi^+_{m-1} - 1$, or equivalently, $1 + x_m \leq \Phi^+_{m-1}$, which is a~consequence of the $\gamma(r)$-dependability of \Det.

That is, the post-ack budget of $\Det$ evolves according to a mapping $z \mapsto \gamma(r) \cdot (z - 1)$, and $1+r = \gamma(r)/(\gamma(r)-1)$ is the fixed point of this mapping. Hence, if the post-ack budget of \Det drops below $1+r$, it remains 
below $1+r$, and moreover, we may show that the gap to the fixed point grows geometrically. This means that the post-ack budget of \Det eventually becomes negative, which contradicts the $\gamma(r)$-dependability of \Det. \cref{lem:test} formalizes this threat: after acknowledging any burst, a $\gamma(r)$-dependable algorithm must keep its post-ack budget at least $1+r$.

\paragraph{Constraints from all bursts.}

By expanding the definition of the post-ack budget, we turn \cref{lem:test} into a family of inequalities which relate the cost of \Det after burst $i$ to its cost after burst $i-1$ (cf.~\cref{lem:certificate}). In particular, for a single burst, $\Phi^+_1 = (\gamma(r)-1) \cdot (1 + x_1)$, and thus $1 + x_1 \geq (1+r)/(\gamma(r)-1) = r \cdot (1+r) = 1 + c$, i.e., the latency $x_1$ of the first burst is at least~$c$. Moreover, if the adversary issues a large second burst at latency $y > x_1$, then \Det pays at least $(1 + r) \cdot (1 + x_1 + r - y)$, up to an arbitrarily small error term.

\paragraph{Adversarial distribution.}

In the second part of the proof (\cref{sec:theorem}), we study the behavior of a deterministic algorithm \Det satisfying the properties above on a random instance. Our instances consist of one or two bursts. With probability $q$, the instance is a single packet; otherwise a second, large burst arrives at a random latency $y \in [c, r]$. We choose the distribution so that the expected competitive ratio of \Det does not depend on $x_1$, and the min-max principle of \cref{lem:minmax} yields \cref{thm:res_tcp_lb}. 

We note that \cref{lem:certificate} provides constraints for instances with any number of bursts, while our distribution uses only two of them. We believe that longer instances are a natural route towards closing the gap to \cref{thm:res_tcp}.

\section{Conclusions}
\label{sec:conclusions}

In this paper, we studied pure tail constraints in competitive analysis, i.e., the tradeoff between the dependability of a randomized online algorithm and its expected competitive ratio. 

For online bidding and line search, we determined the Pareto-optimal frontier of this tradeoff. In both problems, the optimal tradeoff is attained by known randomized algorithms, parametrized by the growth ratio $r$ of consecutive steps. Our main contributions here are the lower bounds: we showed that dependability essentially forces a deterministic algorithm to follow a geometric sequence with ratio $r$. 

For TCP acknowledgment, we gave an algorithm which achieves the tradeoff known for ski rental, and we complemented it with a lower bound that beats the randomized one of $e/(e-1)$ for dependability ratio smaller than $\gamma^\star \approx 2.22$. We do not know where the true frontier for TCP acknowledgment lies between the two curves of \cref{fig:tcp}, and closing this gap is the main problem left open by our work. Our lower bound uses instances with at most two bursts only, while dependability constrains an algorithm also on longer instances, and we believe this leaves room for a stronger bound.

\subsection*{AI Use Statement}

In this work, we used a generative AI tool (ChatGPT 5.6 Sol) for three tasks. 
\begin{itemize}
    \item First, we used it to produce the Python code that plots the dependability-competitiveness tradeoff curves in \cref{fig:bidding_line}. We verified the produced code manually. 
    \item Second, we used it to find the probability distribution over instances used in \cref{sec:theorem}, chosen so that the expected competitive ratio of an algorithm does not depend on the latency at which it acknowledges the first burst. We proved all properties of this distribution manually; the proofs are given in \cref{sec:theorem}. 
    \item Conversations with it were used as an inspiration for some parts of the proof of \cref{lem:tcp_dependable}. The entire proof was written and verified by hand.
\end{itemize}
We did not use generative AI tools for any other task with a required disclosure. We take full responsibility for the final content of this work.

\subsection*{Acknowledgments}

This work has been supported by Polish National Science Centre grants 2022/45/B/ST6/00559 and 2020/39/B/ST6/01641.

\bibliography{references,ref02}
\bibliographystyle{iclr2027_conference}

\appendix

\section{Related Work}
\label{sec:related}

For online bidding, folklore deterministic and randomized algorithms based on geometric scaling achieve competitive ratios of $4$ and $e$, respectively, which were later shown to be optimal by \citet{ChKeNY08}. For line search, \citet{BaCuRa93} proved the optimal deterministic competitive ratio of $9$, while \citet{KaReTa96} established the optimal randomized competitive ratio of approximately $4.591$. The standard deterministic algorithms for online bidding and line search share the same underlying doubling structure; see \citet{chrobak2006sigact} for an explicit connection. This relation, however, does not extend directly to randomized algorithms.

\citet{DoGoSc01} introduced the TCP acknowledgment problem and established the optimal deterministic competitive ratio of $2$. The optimal randomized competitive ratio is $e/(e-1)$: a lower bound was shown by \citet{Seiden00} and an algorithm was given by \citet{KaKeRa03}, who uncovered a connection to ski rental. TCP acknowledgment is the simplest problem in a hierarchy of online aggregation problems. Its two-level extension is the online joint replenishment problem, for which \citet{BuKLMS13} gave a $3$-competitive deterministic algorithm and proved a lower bound of approximately $2.64$; the lower bound was later improved to approximately $2.754$~\citep{BBCDNS15}. Both problems are special cases of the multi-level aggregation problem on trees of depth $D$ that admits $2D$-competitive deterministic algorithms~\citep{AhCKPZ26}, where TCP acknowledgment and joint replenishment correspond to $D = 1$ and $D = 2$, respectively.

Problems studied in this paper have also been studied in the context of learning-augmented online optimization~\citep{PuSvKu18,wei2020optimal,bamas2020primal,anand2021regression,im2023online,AnBiDS24,angelopoulos2024online,shin2025improved,angelopoulos2026learning,cabello2026searching}. In this framework, an online algorithm is given a prediction of unknown quality. A central goal is to characterize the tradeoff between \emph{consistency}, the performance under accurate predictions, and \emph{robustness}, the worst-case performance under arbitrary predictions. The consistency-robustness tradeoff resembles the tradeoff between expected competitiveness and dependability under pure tail constraints, but is of different nature. Consistency and robustness bound the performance of an algorithm under two different qualities of the prediction, while expected competitiveness and dependability constrain it on the same input. 

Finally, we note that while the tradeoffs related to tail risk are rather recent, the idea of bounding the performance of online algorithms with high probability rather than with just the expected value has been studied in literature for a long time. Examples of studied problems include, e.g., call control~\citep{LeMaPR01}, data management strategies~\citep{MaMeVW97}, and online matching~\citep{MihTro24}. Moreover, for some problems, whose algorithms use phase-based strategies and whose costs are potentially unbounded, it is possible to transform the guarantees on the competitive ratio that hold in expectation into high-probability ones \citep{KoKrKM22}. This applies in particular to paging and metrical task systems.

\section{Competitive Ratio and Min-max Principle}
\label{sec:competitive_ratio}

\paragraph{Strict and non-strict competitiveness.}

In the literature, the competitive ratio is sometimes defined in a weaker form, in which the cost of an online algorithm may exceed the required multiple of $\Opt(\sigma)$ by an additive constant which is independent of the instance. In this paper, both notions are \emph{strict}, i.e., we do not allow additive constants. We do this for simplicity only: our upper bounds hold without additive constants, and our lower bounds can be extended to the weaker definition. Namely, our lower bound constructions can be made arbitrarily expensive, which makes any additive term negligible. In online bidding and line search, it suffices to place the hidden value far away, and in TCP acknowledgment, to repeat the construction many times, far apart in time.

\paragraph{Min-max principle.}

The following lemma is a slightly extended and reworded min-max principle, which we use to obtain lower bounds on the competitive ratio of randomized algorithms. 

\minmax*

\begin{proof}
    Let \Rand be a $\rho^\star$-competitive randomized algorithm with $\supp(\Rand) \subseteq \gD^\star$; we will show that $\rho^\star \geq \rho$.
    Recall that \Rand is a probability distribution over deterministic algorithms, and by our assumptions, $\supp(\Rand) \subseteq \gD^\star \subseteq \gD$. Then,
    \begin{align*}
        \rho^\star
        & \geq \sup_\sigma \E_{\Det \sim \Rand}[\ratio(\Det, \sigma)]
            && \text{(by competitiveness of \Rand)}\\
        & \geq \E_{\sigma \sim \pi} \E_{\Det \sim \Rand} [\ratio(\Det, \sigma)] \\
        & = \E_{\Det \sim \Rand} \E_{\sigma \sim \pi}[\ratio(\Det, \sigma)] \\
        & \geq \inf_{\Det \in \supp(\Rand)} \E_{\sigma \sim \pi}[\ratio(\Det, \sigma)] \\
        & \geq \inf_{\Det \in \gD} \E_{\sigma \sim \pi}[\ratio(\Det, \sigma)]
            && \text{(by $\supp(\Rand) \subseteq \gD$)}\\
        & \geq \rho.
            && \text{(by the lemma assumption)}
            \qedhere
    \end{align*}
\end{proof}
\section{Proof of \texorpdfstring{\cref{lem:step_ratio}}{\ref*{lem:step_ratio}} (Step Ratio Lemma)}
\label{sec:step_ratio_proof}

\stepratio*

\paragraph{Notation.}

Throughout this part, we fix sequence $\{x_i\}_{i=0}^n$ and parameters $c$, $r$ and $\eps$ as in the assumptions of \cref{lem:step_ratio}. For the ease of notation, we define $B \define r^2/(r-1)$ and \emph{step ratios} $q_k \define x_k / x_{k-1}$ for $k \geq 1$. Note that $q_k > 0$, but, as the sequence need not be monotone, $q_k$ may be smaller than~$1$. Finally, for $q > 1$, we define 
\[
    g_c(q) \define \frac{c+q}{\ln q}. 
\]
Then, the assumption on $\eps$ reads $\eps < g_c(r)$. Proof of the lemma uses the following simple calculus claim about function $g_c$.

\begin{claim}
    \label{cla:efficiency}
    It holds that $\rho_c > 2$. The function $g_c$ is strictly decreasing on $(1,\rho_c]$ and strictly increasing on $[\rho_c,\infty)$.
\end{claim}

\begin{proof}[Proof of \cref{cla:efficiency}]
    We have $g_c'(q) = (\ln q - (c+q)/q)/(\ln q)^2$, so the sign of $g_c'$ is opposite to the sign of $h_c(q) \define c + q - q \ln q$. By taking the derivative of $h_c$, we get that $h_c$ is strictly decreasing on $(1,\infty)$. Furthermore, $h_c(1) > 0$ and $\lim_{q \to \infty} h_c(q) = -\infty$. Hence, $h_c$ has a unique zero $\rho_c$ in~$(1,\infty)$. As $h_c$ is positive on $(1,\rho_c)$ and negative on $(\rho_c,\infty)$, the function $g_c$ is monotonic on both intervals as claimed. Finally, as $h_c(2) > 0$, it holds that $\rho_c > 2$.
\end{proof}

\paragraph{Per-step bounds.}

Suppose we would like to prove \cref{lem:step_ratio} by a straightforward induction (setting for simplicity $\eps = 0$ and $F = 0$). Then, the inductive step would add $c + q_k$ to the left hand side of the desired inequality and $g_c(r) \cdot \ln q_k$ to the right hand side. For $q_k = r$, the two terms are equal. For $q_k < r$, the left hand side is at least as large, as $g_c(r) > 0$ and $g_c$ is decreasing on $(1,\rho_c]$ by \cref{cla:efficiency}, so the inductive step would succeed also in this case. We may even allow $q_k$ slightly larger than~$r$, at the expense of replacing $g_c(r)$ by $g_c(r) - \eps$ (cf.~\cref{lem:per_step}).
    
Hence, the hard part is to analyze the case $q_k > r$. This is where the upper bound on $S_k$ from the lemma assumption comes into play: we show that it implies that, from some point on, the step ratios cannot exceed~$r$ by more than a small term. From this point on, we may apply \cref{lem:per_step}, and we bound the initial terms by additive constant $F$.

\begin{lemma}
    \label{lem:per_step}
    Let $\{z_k\}_{k \geq 1}$ be a sequence converging to $r$, and assume that $q_k \leq z_k$ for every $k \geq 1$. Then, there exists~$K$, depending on $c$, $r$, $\eps$ and the sequence $\{z_k\}$ only, such that
    \[
        c + q_k \geq (g_c(r) - \eps) \cdot \ln q_k \qquad \text{for every $k > K$}.
    \]
\end{lemma}

\begin{proof}
    As $r \leq \rho_c$, \cref{cla:efficiency} implies that $g_c$ is decreasing on $(1,r]$, and thus $g_c(q) \geq g_c(r)$ for every $q \in (1,r]$. Moreover, $g_c$ is continuous at $r$. Hence, there exists small enough $\theta > 0$, such that
    \begin{equation}
        \label{eq:eff}
        g_c(q) \geq g_c(r) - \eps \qquad \text{for every $q \in (1, r+\theta]$} .
    \end{equation}
    As the sequence $\{z_k\}$ converges to $r$, we may fix a $K$, such that $z_k \leq r + \theta$ for every $k > K$.

    Now fix $k > K$; then $q_k \leq z_k \leq r + \theta$. If $q_k \leq 1$, then $\ln q_k \leq 0$, and thus $c + q_k > 0 \geq (g_c(r) - \eps) \cdot \ln q_k$, where the latter inequality follows by $\eps < g_c(r)$. Otherwise $q_k \in (1, r+\theta]$, and then $c + q_k = g_c(q_k) \cdot \ln q_k \geq (g_c(r) - \eps) \cdot \ln q_k$ by the definition of $g_c$ and \eqref{eq:eff}.
\end{proof}

\paragraph{Recurrence.}

The following function $\varphi(t) \define B/(B-t)$ for $t \in [0,B)$ and the recurrence
\begin{equation}
    \label{eq:tau_recursion}
    \tau_0 \define 0, \qquad \tau_{k+1} \define \varphi(\tau_k) \quad \text{ for $k \geq 0$}
\end{equation}
will be later used to track the evolution of the sequence $\{ S_k/x_k \}_{k \geq 0}$. 

\begin{claim}
    \label{cla:recursion}
    The function $\varphi$ is increasing and the sequence $\{ \tau_k \}_{k \geq 0}$ converges to $\frac{B}{r} = \frac{r}{r-1}$.
\end{claim}

\begin{proof}
    Since the derivative of $\varphi$ is $\varphi'(t) = B/(B-t)^2 > 0$, the function $\varphi$ is increasing. 
    
    The fixed points of $\varphi$ satisfy $B/(B-t) = t$, or equivalently, $t^2 - B \cdot t + B = 0$. That is, the fixed points are $B/r$ and $r$, where $B/r$ is the smaller one of them, as $B/r = r/(r-1) \leq r$.

    For $t \in [0, B/r)$ we have $t^2 - B \cdot t + B > 0$, or equivalently $\varphi(t) = B/(B-t) > t$. Thus, the sequence $\{\tau_k\}_{k \geq 0}$ is strictly increasing as long as it stays below $B/r$. It indeed stays there, as $\varphi$ is increasing and $\varphi(B/r) = B/r$. Thus, it converges to some fixed point $t^* \leq B/r$, and the only fixed point with this property is $B/r$. 
\end{proof}

\begin{proof}[Proof of \cref{lem:step_ratio}]
    We start with applying the recurrence above to analyze how the sequence $\{ S_k/x_k \}_{k \geq 0}$ evolves. Namely, by induction on $k$, we show that 
    \begin{equation}
        \label{eq:S_k_over_x_k}
        S_k/x_k \geq \tau_k \qquad \text{for every $k \geq 0$}.
    \end{equation}
    For $k = 0$, the inequality \eqref{eq:S_k_over_x_k} holds as $S_0 / x_0 = 0 = \tau_0$. For the induction step, fix $k \geq 1$ and suppose that \eqref{eq:S_k_over_x_k} holds for $k-1$. Then, 
    \begin{align*}
        \frac{S_k}{x_k} = 1 + \frac{S_{k-1}}{x_k}
        & \geq 1 + \frac{S_{k-1}}{B \cdot x_{k-1} - S_{k-1}}
            && \text{(by $x_k + S_{k-1} = S_k \leq B \cdot x_{k-1}$)} \\
        & = \frac{B \cdot x_{k-1}}{B \cdot x_{k-1} - S_{k-1}} = \varphi \left( \frac{S_{k-1}}{x_{k-1}} \right) \\
        & \geq \varphi(\tau_{k-1}) = \tau_{k},
            && \text{(as $\varphi$ is increasing by \cref{cla:recursion})} 
    \end{align*}
    which concludes the induction step and proves \eqref{eq:S_k_over_x_k}.

    Next, we bound the step ratios. We set $z_k \define B - \tau_{k-1}$ for $k \geq 1$. By the lemma assumption, $B \geq S_k/x_{k-1} = (x_k + S_{k-1})/x_{k-1} = q_k + S_{k-1}/x_{k-1}$, so by reorganizing the terms and applying~\eqref{eq:S_k_over_x_k}, we obtain
    \begin{equation}
        \label{eq:cap}
        q_k \leq B - \frac{S_{k-1}}{x_{k-1}} \leq B - \tau_{k-1} = z_k \qquad \text{for every $k \geq 1$} .
    \end{equation}
    By \cref{cla:recursion}, the sequence $\{z_k\}_{k \geq 1}$ converges to $B - B/r = r$. Thus, \cref{lem:per_step} applied to $\{z_k\}$ yields $K$, such that
    \begin{equation}
        \label{eq:per_step}
        c + q_k \geq (g_c(r) - \eps) \cdot \ln q_k \qquad \text{for every $k > K$}.
    \end{equation}
    Moreover, by \eqref{eq:cap}, $q_i \leq B$ for every $i$, and thus for any $k$ it holds that 
    \begin{equation}
        \label{eq:per_step_2}
        \ln x_k = \sum_{i=1}^{k} \ln q_i \leq k \cdot \ln B.
    \end{equation}
    Next, we set $F \define (g_c(r) - \eps) \cdot K \cdot \ln B$, and we consider two cases. If $n \leq K$, then $\ln x_n \leq K \cdot \ln B$. In such case, the right hand side of the inequality from the lemma is $(g_c(r) - \eps) \cdot \ln x_n - F \leq (g_c(r) - \eps) \cdot K \cdot \ln B - F = 0$, and the lemma follows. Thus, we may assume that $n > K$, and then
    \begin{align*}
        c \cdot n + \sum_{k=1}^{n} q_k
        & \geq \sum_{k=K+1}^{n} (c + q_k) \\
        & \geq (g_c(r) - \eps) \cdot \sum_{k=K+1}^{n} \ln q_k
            && \text{(by \eqref{eq:per_step})} \\
        & = (g_c(r) - \eps) \cdot (\ln x_n - \ln x_K) \\
        & \geq (g_c(r) - \eps) \cdot (\ln x_n - K \cdot \ln B)
            && \text{(by \eqref{eq:per_step_2})} \\
        & = (g_c(r) - \eps) \cdot \ln x_n - F . 
            && \qedhere
    \end{align*}
\end{proof}

\section{Proofs for \texorpdfstring{\cref{sec:bid}}{\ref*{sec:bid}} (Online Bidding)}
\label{sec:bid_proofs}

Throughout this section, we fix $r \in [2, e]$ and let $\gamma(r) \define r^2/(r-1)$.

\bidub*

\begin{proof}
    Let $T \geq 1$ be the threshold, $\tau = \log_r(T)$, and let $b$ be the largest bid paid by the algorithm, defined by $b/r < T \leq b$. The value $\log_r(b)$ is distributed uniformly in $[\tau, \tau+1)$. Thus, we may view the algorithm as producing bids ending at $r^{\zeta + \tau}$ for $\zeta \define \log_r(b) - \tau$ uniform in $[0,1)$, each bid equal to $r$ times the preceding one. By bounding this finite sequence by the corresponding infinite one, we obtain, for any fixed value of $\zeta$,
    \begin{equation}
    \label{eq:bid_r_cost}
        \ratio(\Bid_r(\xi), T) 
        \leq \frac{1}{T} \cdot \sum_{i=0}^{\infty} r^{\zeta + \tau - i} 
        =r^{\zeta} \cdot \frac{r^{\tau}}{T} \cdot \sum_{i=0}^{\infty} r^{-i} 
        =r^{\zeta} \cdot \frac{r}{r-1}. 
    \end{equation}
    To bound the expected competitive ratio of $\Bid_r$, we first compute the expectation of $r^\zeta$ as $\E_{\zeta \sim [0,1)}[r^{\zeta}] = \int_0^1 r^z \, dz = (r-1)/\ln r$. Thus,
    \[
        \E_{\xi \sim [0,1)}[\ratio(\Bid_r(\xi), T)] 
        \leq \frac{r}{r-1} \cdot \E_{\zeta \sim [0,1)}\left[r^{\zeta}\right] 
        = \frac{r}{\ln r}.
    \]
    Next, we analyze the dependability of $\Bid_r$. By the definition of $\zeta$, we have $r^\zeta < r$. Thus, by \eqref{eq:bid_r_cost}, it holds that $\ratio(\Bid_r(\xi), T) < r \cdot r/(r-1) = \gamma(r)$ regardless of the choice of $\xi$.
\end{proof}

\detbid*

\begin{proof}
    We assume that $\delta < 4r/\ln r$, as otherwise the lemma holds trivially. Let $F(\eps)$ be the parameter (depending on $\eps$) resulting from \cref{lem:step_ratio}, when it is applied with $c = 0$ and $\eps$. We choose a real number $U > 1$ large enough so that 
    \begin{equation}
        \label{eq:choice_of_U}
        \frac{1}{\ln U} \leq \frac{\delta}{4} \cdot \min \left\{ 
            \frac{\ln r}{r},\;
            \frac{1}{F(\delta/4)},\;
            \frac{e \cdot (r-1)}{r^2}
        \right\} .
    \end{equation}
    Distribution $\pi(\delta)$ is then defined as follows: for any $t \in [1, U]$, the density of the threshold at $t$ is $v_t = 1/(t \cdot \ln U)$. Since the algorithm knows $\pi(\delta)$, we can assume without loss of generality that it never bids a number outside the range $[1, U]$.

    Fix a deterministic $\gamma(r)$-dependable algorithm $\Det$. Since $\Det$ is deterministic, bids of $\Det$ are a~sequence of numbers $x_1 < x_2 < \ldots$. To simplify notation, we set $x_0 = 1$. Define $n$ as an index of the bid satisfying
    \[
        x_{n-1} < U / e \leq x_{n}. 
    \]
    Such an index must exist, as otherwise $\Det$ would not be competitive for thresholds greater than~$U/e$. We charge $\Det$ only for the bids $x_1, \dots, x_n$, and we ignore its cost of any subsequent bids. 
    Next, let $S_k \define \sum_{i=1}^{k} x_i$ denote the total cost \Det pays on the first $k$ bids. 
    This gives the following bound
    \begin{align}
        \E_{\sigma \sim \pi(\delta)}[\ratio(\Det, \sigma)] 
        \nonumber
        &\geq \sum_{i=1}^n \left(\int_{x_{i-1}}^{x_i} \frac{S_i}{t} \cdot v_t \, dt \right) + \int_{x_n}^{U} \frac{S_n}{t} \cdot v_t \, dt\\
        \nonumber
        &=\sum_{i=1}^n \left(\int_{x_{i-1}}^{x_i} \frac{\sum_{j=1}^i x_j}{t^2 \cdot\ln U} \, dt\right) + \int_{x_n}^{U} \frac{\sum_{j=1}^n x_j}{t^2 \cdot \ln U} \, dt\\
        \nonumber
        &=\frac{1}{\ln U} \cdot \sum_{i=1}^n  x_i\int_{x_{i-1}}^{U} \frac{1}{t^2} \, dt\\
        \nonumber
        &=\frac{1}{\ln U} \cdot \sum_{i=1}^n  x_i \left( \frac{1}{x_{i-1}} - \frac{1}{U}\right)\\
        \label{eq:det_bid_cost}
        &=\frac{1}{\ln U} \cdot \sum_{i=1}^n  \frac{x_i} {x_{i-1}} - \frac{S_n}{U \ln U} .
    \end{align}

    As \Det is $\gamma(r)$-dependable, its cost on the first $k$ bids, $S_k$, cannot exceed $\gamma(r)$ times the cost of \Opt, which can be arbitrarily close to $x_{k-1}$. Hence 
    \begin{align}
        \label{eq:prefix_sum_bound}
        S_k 
        & \leq \gamma(r) \cdot x_{k-1} = \frac{r^2}{r-1} \cdot x_{k-1} \qquad \text{for every $k \geq 1$},
    \intertext{and thus also for $k = n$,}
        \label{eq:prefix_sum_bound_n}
        S_n 
        & \leq \frac{r^2}{r-1} \cdot x_{n-1} \leq \frac{r^2}{r-1} \cdot \frac{U}{e}.
    \end{align}
    By \eqref{eq:prefix_sum_bound}, we may apply \cref{lem:step_ratio} with $c = 0$ and $\eps = \delta/4$ obtaining
    \begin{equation}
        \label{eq:step_ratio_bound_bid}
        \sum_{i=1}^{n} \frac{x_i}{x_{i-1}} \geq \left( \frac{r}{\ln r} - \frac{\delta}{4} \right) \cdot \ln x_n - F(\delta/4),
    \end{equation}
    which plugged into \eqref{eq:det_bid_cost} gives 
    \begin{align*}
        \E_{\sigma \sim \pi(\delta)}[\ratio(\Det, \sigma)] 
        &\geq \frac{\ln x_n}{\ln U} \cdot \left( \frac{r}{\ln r} - \frac{\delta}{4} \right) 
        - \frac{F(\delta/4)}{\ln U} - \frac{S_n}{U \ln U} \\
        & \geq \left(1 - \frac{1}{\ln U} \right) \cdot \left( \frac{r}{\ln r} - \frac{\delta}{4} \right) 
        - \frac{F(\delta/4)}{\ln U} - \frac{r^2}{e \cdot (r-1) \cdot \ln U} \\
        & \geq \frac{r}{\ln r} - \frac{\delta}{4} 
        - \frac{r}{\ln r \cdot \ln U} - \frac{F(\delta/4)}{\ln U} - \frac{r^2}{e \cdot (r-1) \cdot \ln U} \\
        & \geq \frac{r}{\ln r} - \delta . 
            \qedhere
    \end{align*}
\end{proof}

We can now prove \cref{thm:res_bidding}, restated below.

\thmbid*

\begin{proof}
    The upper bound is \cref{lem:bid_ub}. By combining \cref{lem:det_bid} with the min-max principle (\cref{lem:minmax}), we obtain that no randomized $\gamma(r)$-dependable algorithm can achieve expected competitive ratio better than $r/\ln r - \delta$. As $\delta$ can be arbitrarily small, the theorem follows.
\end{proof}

\section{Proofs for \texorpdfstring{\cref{sec:line}}{\ref*{sec:line}} (Line Search)} 
\label{sec:line_proofs}

Let $r_0$ be the solution of $r_0 = (1 + r_0)/\ln r_0$. Throughout this section, fix $r \in [2, r_0]$ and let $\gamma(r) \define 1 + (2r^2)/(r-1)$.

The upper bound of \cref{thm:res_line} is attained by the following algorithm called \smart, defined in \citet{KaReTa96}. It is parametrized by a number $r \in (1, +\infty)$.

Algorithm $\smart_r$: first, choose a value $\xi \in [0,1)$ uniformly at random and a random starting direction. Then move to points $r^{\xi}, r^{\xi+1}, r^{\xi+2}, \ldots$ in the alternating directions, until the target is found.

\begin{lemma}[{\citealp[Theorem 3.1]{KaReTa96}}] \label{lem:exp_smart}
    Fix $r > 1$. Algorithm $\smart_r$ has expected competitive ratio $1 + \frac{1 + r}{\ln r}$.
\end{lemma}

This competitive ratio is minimized when $r = (1 + r) /\ln r$, which holds for value $r_0 \approx 3.59112$. Then the expected ratio is $1 + r_0 \approx 4.59112$.

\begin{restatable}{lemma}{linedep} \label{lem:line_dep}
    Fix $r \in [2, r_0]$. Algorithm $\smart_r$ is $(1 + \frac{2r^2}{r-1})$-dependable.
\end{restatable}

\begin{proof}
    Fix the value of $\xi$ and let $k$ be an integer such that
    \[ 
        r^{\xi+k} < T \leq r^{\xi+k+1}. 
    \]
    The cost of $\smart_r$ depends on how the chosen starting direction aligns with the direction of~$T$. If the cow travels to $r^{k+\xi}$ for the first time in the direction opposite to~$T$, then the total cost is $2 \cdot \sum_{i=0}^k r^{\xi+i} + T.$ Otherwise, after visiting $r^{k+\xi}$, it then visits $r^{k+1+\xi}$ in the direction opposite to~$T$ and only then comes back to reach $T$. Its cost is then $2 \cdot \sum_{i=0}^{k+1} r^{\xi+i} + T.$

    Clearly, the second case corresponds to the higher cost. This cost can be bounded as
    \begin{align*}
        2 \cdot \sum_{i=0}^{k+1} r^{\xi+i} + T &\leq  2 \cdot \frac{r^\xi(r^{k+2} - 1)}{r-1} + T \\
        &\leq \frac{2r^2 \cdot r^{\xi+k}}{r-1} + T \\
         &< \frac{2r^2}{r-1} \cdot T + T .
    \end{align*}
    Since the optimal cost is $T$, this gives us the dependability bound.
\end{proof}

Combining \cref{lem:exp_smart} with \cref{lem:line_dep} shows that $\smart_r$ gives the upper bound part of \cref{thm:res_line} for any $r \in [2, r_0]$. To prove a tight lower bound, we need the following lemma.

\detline*

\paragraph{Notation and probability distribution.} 
We assume that $\delta < 4 \cdot (1+r)/\ln r$, as otherwise the lemma holds trivially. Let $F(\eps)$ be the parameter (depending on $\eps$) resulting from \cref{lem:step_ratio}, when it is applied with $c = 1$ and $\eps$. We choose a real number $U > 1$ large enough so that 
\begin{equation*}
    \frac{1}{\ln U} \leq \frac{\delta}{4} \cdot \min \left\{
        \frac{\ln r}{1 + r},\;
        \frac{1}{F(\delta/4)},\;
        \frac{e \cdot (r-1)}{2 \cdot r^2}
    \right\} .
\end{equation*}
Distribution $\pi(\delta)$ is then defined as follows: for any $t \in [1, U]$ on any of the two paths, the density of the target at $t$ is $v_t = 1/(2t \cdot \ln U)$. Note that the distribution is symmetric on both paths. Since the algorithm knows $\pi(\delta)$, we can assume without loss of generality that it never travels further than~$U$ from the starting point.

Fix a deterministic $\gamma(r)$-dependable algorithm $\Det$. It can be described by 
\begin{itemize}
    \item a sequence of numbers $x_1, x_2, \ldots$ from $[1,U]$ satisfying
        $x_1 < x_3 < x_5 < \ldots$ and $x_2 < x_4 < x_6 < \ldots$, and
    \item a starting direction.
\end{itemize}
$\Det$ first moves to points $x_1, x_2, x_3, \ldots$ on alternating paths, starting with $x_1$ in the chosen starting direction, until it finds the target. Any algorithm that works differently can be easily transformed into such an algorithm without increasing its cost. For succinctness, we set $x_{-1} = x_0 = 1$. Define $n$ as an index satisfying
\[
    x_{n-1} < U/e \leq x_{n}.
\]
Such an index must exist, as otherwise $\Det$ would not be competitive for targets further than $U/e$. If after reaching the point $x_n$, the searcher has not yet reached the target, it moves back to $0$ and is not required to walk any further. The cost of \Det is then increased by the position of the target, which lower-bounds the cost the algorithm would pay from that point. All further costs of \Det are forgiven.

Let $S_k \define \sum_{i=1}^{k} x_i$. Then $2 \cdot S_k$ is the total cost of \Det in the first $k$ steps.

\begin{lemma}\label{lem:comp_line}
    Fix $\delta > 0$ and let distribution $\pi(\delta)$ be defined as above. Then,
    \[
        \E_{\sigma \sim \pi(\delta)}[\ratio(\Det, \sigma)] \geq 1 + \frac{1}{\ln U} \left(n + \sum_{i=1}^n \frac{x_i}{x_{i-1}} \right) - \frac{2\cdot S_n}{U \ln U}.
    \]
\end{lemma}

\begin{proof}
Let $t$ denote the position of the target. For every $i \geq 1$, if the target is positioned on the same path as $x_i$ and in the interval $(x_{i-2}, x_i]$, then \Det pays exactly $2 \cdot S_{i-1} + t$. Furthermore, if the target is placed on the same path as $x_n$ and $t > x_n$, due to forgiveness mentioned above, the algorithm pays $2 \cdot S_n + t$. Symmetrically, if it is placed above $x_{n-1}$ on the other path, \Det also pays $2 \cdot S_n + t$. Therefore, its expected ratio is at least
\begin{align*}
    \E_{\sigma \sim \pi(\delta)}&[\ratio(\Det, \sigma)]  \\
        & \geq \sum_{i=1}^n \left(\int_{x_{i-2}}^{x_i} \frac{2 S_{i-1}+t}{t} \cdot v_t \, dt \right) + \int_{x_n}^{U} \frac{2 S_n+t}{t} \cdot v_t \, dt + \int_{x_{n-1}}^{U} \frac{2 S_n+t}{t} \cdot v_t \, dt \\
        &= 1 + \sum_{i=1}^n \left(\int_{x_{i-2}}^{x_i} \frac{2 \cdot S_{i-1}}{t} \cdot v_t \, dt \right) + \int_{x_n}^{U} \frac{2 \cdot S_n}{t} \cdot v_t \, dt + \int_{x_{n-1}}^{U} \frac{2 \cdot S_n}{t} \cdot v_t \, dt \\
        &= 1 + \sum_{i=1}^n \left(\int_{x_{i-2}}^{x_i} \frac{\sum_{j=1}^{i-1} x_j}{t^2\ln U} \, dt \right) + \int_{x_n}^{U} \frac{\sum_{j=1}^{n} x_j}{t^2\ln U} \, dt + \int_{x_{n-1}}^{U} \frac{\sum_{j=1}^{n} x_j}{t^2\ln U} \, dt
\end{align*}
By changing the order of summation and, for every $x_i$, summing all integrals in which $x_i$ appears, we obtain
\begin{align*}
\E_{\sigma \sim \pi(\delta)}[\ratio(\Det, \sigma)] 
    &\geq 1 + \frac{1}{\ln U} \cdot \sum_{i=1}^n  x_i\left(\int_{x_{i}}^{U} \frac{1}{t^2} \, dt + \int_{x_{i-1}}^{U} \frac{1}{t^2} \, dt\right)\\
    &=1 + \frac{1}{\ln U} \cdot \sum_{i=1}^n  x_i \left(\frac{1}{x_{i}} - \frac{1}{U} + \frac{1}{x_{i-1}} - \frac{1}{U}\right)\\
    &=1 + \frac{1}{\ln U} \cdot \left(n + \sum_{i=1}^n \frac{x_i}{x_{i-1}} \right) - \frac{2 \cdot S_n}{U \ln U}. 
    \qedhere
\end{align*}
\end{proof}

\begin{proof}[Proof of \cref{lem:det_line}] 
As \Det is $\gamma(r)$-dependable, its cost of reaching any $x \in (x_{k-1}, x_{k+1}]$, which is $x + 2 \cdot S_k$, cannot exceed $\gamma(r)$ times the cost of \Opt, which is $x$. As $x$ can be arbitrarily close to~$x_{k-1}$, 
\begin{equation*}
    x_{k-1} + 2\cdot S_k \leq \left(1 + \frac{2r^2}{r - 1}\right) \cdot x_{k-1} \qquad \text{for every $k \geq 1$}.
\end{equation*}
Equivalently,
\begin{align}
    \label{eq:prefix_sum_line}
    S_k &\leq \frac{r^2}{r - 1} \cdot x_{k-1} \qquad \text{for every $k \geq 1$},
    \intertext{and thus also for $k = n$,}
        S_n 
        & \leq \frac{r^2}{r-1} \cdot x_{n-1} \leq \frac{r^2}{r-1} \cdot \frac{U}{e}. \nonumber
\end{align}
By \eqref{eq:prefix_sum_line}, we may apply \cref{lem:step_ratio} with $c = 1$ and $\eps = \delta/4$, obtaining
\begin{equation}
    \label{eq:step_ratio_bound_line}
    n + \sum_{i=1}^n \frac{x_i}{x_{i-1}} \geq \left(\frac{1+r}{\ln r} - \frac{\delta}{4} \right) \cdot \ln x_n - F(\delta/4).
\end{equation}
Finally, we combine \eqref{eq:step_ratio_bound_line} with \cref{lem:comp_line} to obtain
\begin{align*}
    \E_{\sigma \sim \pi(\delta)}[\ratio(\Det, \sigma)] &\geq 1 + \frac{1}{\ln U} \cdot \left(n + \sum_{i=1}^n \frac{x_i}{x_{i-1}} \right) - \frac{2 \cdot S_n}{U \ln U}\\
    &\geq 1 + \frac{\ln x_n}{\ln U} \cdot \left(\frac{1+r}{\ln r} - \frac{\delta}{4} \right) - \frac{F(\delta/4)}{\ln U} - \frac{2 \cdot S_n}{U \ln U}\\
    &\geq 1 + \left(1 - \frac{1}{\ln U}\right) \cdot \left(\frac{1+r}{\ln r} - \frac{\delta}{4}\right) - \frac{F(\delta/4)}{\ln U} - \frac{2 \cdot r^2}{e \cdot (r-1) \cdot \ln U}\\
    &\geq 1 + \frac{1+r}{\ln r} - \frac{\delta}{4} - \frac{1 + r}{\ln r \cdot \ln U} - \frac{F(\delta/4)}{\ln U} - \frac{2 \cdot r^2}{e \cdot (r-1) \cdot \ln U}\\
    &\geq 1 + \frac{1+r}{\ln r} - \delta. \qedhere
\end{align*}
\end{proof}

 We are ready to prove \cref{thm:res_line}, restated below.

\thmline*

\begin{proof}
    The upper bound comes from combining \cref{lem:exp_smart} with \cref{lem:line_dep}. By combining \cref{lem:det_line} with the min-max principle (\cref{lem:minmax}), we obtain that no randomized $\gamma(r)$-dependable algorithm can achieve expected competitive ratio better than $1 + (1 + r)/\ln r - \delta$. As $\delta$ can be arbitrarily small, the theorem follows.
\end{proof}

\section{Proofs for \texorpdfstring{\cref{sec:ub_tcp}}{Section \ref*{sec:ub_tcp}} (TCP Acknowledgment Upper Bound)}
\label{sec:app_tcpub}

\subsection{Dependability}

\dependability*

\begin{proof}
    Consider a deterministic algorithm $\Det_z$ for some $z \in (0, 1]$ and an instance $\sigma$. Let $n_z(\sigma)$ and $W_z(\sigma)$ denote the number of acknowledgments and total latency of $\Det_z$ on $\sigma$, respectively.
    We show the following two inequalities.
    \begin{align}
        z \cdot n_z(\sigma) &\leq \Opt(\sigma) \label{eq:tcp_dep1}\\
        W_z(\sigma) &\leq \Opt(\sigma) \label{eq:tcp_dep2}
    \end{align}
    By combining the two inequalities, we have
    \[
        \Det_z(\sigma)
        = n_z(\sigma) + W_z(\sigma) 
        \leq \left(1 + \frac{1}{z}\right) \cdot \Opt(\sigma),
    \]
    completing the proof.

    We prove \eqref{eq:tcp_dep1} by showing that, for every acknowledgment of $\Det_z$, we can disjointly charge at least $z$ to the optimal solution. 
    Consider the $i$-th acknowledgment of $\Det_z$ at time $T_i$ for $i \in \mathbb{Z}_{> 0}$. 
    If the optimal solution performs an acknowledgment in time interval $(T_{i-1}, T_i]$, we charge $z \leq 1$ to this acknowledgment of the optimal solution for the $i$-th acknowledgment of $\Det_z$. Otherwise, if the optimal solution does not perform any acknowledgment in interval $(T_{i-1}, T_i]$, observe that every packet that has arrived in interval $(T_{i-1}, \tau_i]$ is unacknowledged until time $T_i$ in the optimal solution, implying that the optimal solution incurs latency at least 
    \[
        P(T_{i-1}, \tau_i) \cdot (T_i - \tau_i) = z
    \]
    in interval $(T_{i-1}, T_i]$, where the equality is due to \eqref{eq:tcp:detzdef} from the definition of $\Det_z$. Hence, we can charge $z$ to the latency incurred by the optimal solution during $(T_{i-1}, T_i]$ for the $i$-th acknowledgment of $\Det_z$. Notice that every fraction of the optimal cost $\Opt(\sigma)$ is charged at most once, completing the proof of \eqref{eq:tcp_dep1}.
    
    We now turn to proving \eqref{eq:tcp_dep2}. Let $n_{\Opt}(\sigma)$ and $W_{\Opt}(\sigma)$ respectively denote the number of acknowledgments and latency of the optimal solution, i.e., $\Opt(\sigma) = n_{\Opt}(\sigma) + W_{\Opt}(\sigma)$. It thus suffices to show
    \begin{equation} \label{eq:tcp_dep2_alt}
        W_z(\sigma) - W_{\Opt}(\sigma) \leq n_{\Opt}(\sigma).
    \end{equation}
    Let $K_{\Opt}$ denote the set of all acknowledgment times of the optimal solution, and for any $t \in K_{\Opt}$, let $R_t$ denote the set of all packets acknowledged at time $t$ by the optimal solution. We remark that $\{R_t\}_{t \in K_{\Opt}}$ partitions the entire set of input packets. For each packet $j \in R_t$, let $d_j$ denote the time at which packet $j$ is acknowledged by $\Det_z$. Let $R'_t \subseteq R_t$ denote the set of packets in $R_t$ that are acknowledged by $\Det_z$ later than by the optimal solution, i.e., $R'_t \define \{ j \in R_t : d_j > t \}$.
    
    We claim that, for any acknowledgment time $t \in K_{\Opt}$ of the optimal solution,
    \begin{equation} \label{eq:tcp_dep2_piece}
        \sum_{j \in R'_t} (d_j - t) \leq z.
    \end{equation}
    Observe that, if this claim is true, we immediately derive \eqref{eq:tcp_dep2_alt} since
    \begin{align*}
        W_z(\sigma) - W_{\Opt}(\sigma)
        & = \sum_{t \in K_{\Opt}} \sum_{j \in R_t} (d_j - t)
            \leq \sum_{t \in K_{\Opt}} \sum_{j \in R'_t} (d_j - t) \\
        & \leq \sum_{t \in K_{\Opt}} z
            \leq z \cdot n_{\Opt}(\sigma)
            \leq n_{\Opt}(\sigma),
    \end{align*}
    where the equality comes from the fact that $\{R_t\}_{t \in K_{\Opt}}$ partitions the entire set of input packets.

    To prove the claim, we assume that $R'_t$ is non-empty, as otherwise the claim is trivial. 
    Observe first that all packets in $R'_t$ are acknowledged at the same time, say $T_i$, by $\Det_z$ because they are all unacknowledged at time $t$ in the execution of $\Det_z$.
    Moreover, $R'_t$ is a subset of the packets unacknowledged by $\Det_z$ at time $t$, yielding that $T_{i-1} \leq t$ and $|R'_t| \leq P(T_{i-1}, t)$.
    We can thus bound the left-hand side of \eqref{eq:tcp_dep2_piece} as
    \[
        \sum_{j \in R'_t} (d_j - t) = |R'_t| \cdot (T_i - t) \leq P(T_{i-1}, t) \cdot (T_i - t).
    \]
    Suppose toward contradiction that $\sum_{j \in R'_t} (d_j - t) > z$.
    We then have $P(T_{i-1}, t) \cdot (T_i - t) > z$, implying that $\Det_z$ would have performed the $i$-th acknowledgment strictly earlier than $T_i$ by definition.
    This yields a contradiction to the choice of $T_i$, completing the proof of the claim, and hence the entire proof.
\end{proof}

\begin{remark}
The proof of \cref{lem:tcp_dependable} actually shows $\gamma(z)$-dependability not only for $\Det_z$, but also for any (deterministic or randomized) algorithm such that, for any $i \in \mathbb{Z}_{> 0}$, the $i$-th acknowledgment time $T_i$ is determined as the first time for which there exists $\tau_i \in [T_{i-1}, T_i]$ satisfying
\[
    P(T_{i-1}, \tau_i) \cdot (T_i - \tau_i) = z_i,
\]
where $T_0 \define 0$, and $z_i \in [z, 1]$ is chosen at time $T_{i-1}$ using any information revealed so far.
\end{remark}

\subsection{Competitiveness}

In our analysis of $\Rand_r$, we use the following properties shown by \citet{KaKeRa03}.
\begin{lemma}
    \label{lem:KaKeRa03_properties}
    Let $n_z(\sigma)$ denote the number of acknowledgments made by $\Det_z$ on input $\sigma$. Then,
    \begin{align}
        \Det_z(\sigma) &\leq \Opt(\sigma) + n_z(\sigma) - \int_z^1 n_w(\sigma) \, dw; \label{eq:tcp_Dz_bound}\\
        \Opt(\sigma) &\geq \int_0^1 n_w(\sigma) \, dw. \label{eq:tcp_opt_bound}
    \end{align}
\end{lemma}

Indeed, \eqref{eq:tcp_Dz_bound} can be derived by combining (1) and Lemma 4 of \citet{KaKeRa03}, while \eqref{eq:tcp_opt_bound} is Corollary 5 of the same paper.

We are now ready to prove \cref{lem:tcp_ub}, restated below for convenience.
Recall that $\Rand_r$ initially draws a value $z \in [r, 1]$ from the mixed probability distribution that has
    \begin{itemize}
        \item a discrete mass of $q_r \define \frac{r \cdot e^r}{(r+1) \cdot e - e^r}$ at $z = r$, and
        \item a density $p(z) \define \frac{r+1}{(r+1) \cdot e-e^r} \cdot e^z$ for any $z \in (r, 1]$,
    \end{itemize}
and then executes $\Det_z$ with the sampled $z$.

\tcpublem*

\begin{proof}
    Consider an input $\sigma$. The expected cost incurred by $\Rand_r$ is written as
    \[
        \E[\Rand_r(\sigma)] = q_r \cdot \Det_r(\sigma) + \int_r^1 p(z) \cdot \Det_z(\sigma) \, dz.
    \]
    For any $z \in [r, 1]$, let $\Det_z'(\sigma) = n_z(\sigma) - \int_z^1 n_w(\sigma) \, dw$.
    Using \eqref{eq:tcp_Dz_bound}, we can bound the expected cost by
    \[
        \E\left[\Rand_r(\sigma)\right]
        \leq \Opt(\sigma) + q_r\cdot\Det_r'(\sigma) + \int_r^1 p(z)\cdot\Det_z'(\sigma) \, dz.
    \]
    Note that the last term can be rephrased as follows:
    \begin{align*}
    \int_r^1 p(z)\cdot\Det_z'(\sigma) \, dz
    &= \int_r^1 p(z)\cdot\left(n_z(\sigma) - \int_z^1 n_w(\sigma) \, dw\right) \, dz\\
    &
    \overset{(a)}{=} \int_{r}^1 \left( p(z)\cdot n_z(\sigma) - n_z(\sigma) \cdot \int_{r}^z p(w) \, dw \right) \, dz\\
    &
    \overset{(b)}{=} \int_{r}^1 \Big( p(z)\cdot n_z(\sigma) - n_z(\sigma) \cdot \big(p(z)-p(r) \big) \Big) \, dz\\
    &= p(r)\cdot \int_{r}^1  n_z(\sigma) \, dz,
    \end{align*}
    where $(a)$ follows from changing the order of integration, and $(b)$ from the definition of $p(\cdot)$.
    
    We moreover remark that the number of acknowledgments $n_z(\sigma)$ by $\Det_z$ is non-increasing as $z$ increases. To see this, fix $z$ and $z'$ with $z \leq z'$, and consider two consecutive acknowledgments of $\Det_{z'}$, where the start of the execution counts as an acknowledgment; let $T'_{i-1}$ and $T'_i$ denote the times of these acknowledgments of $\Det_{z'}$. Note that $\Det_z$ must acknowledge at least once in $(T'_{i-1}, T'_i]$ since otherwise all packets unacknowledged by $\Det_{z'}$ at time $t \in (T'_{i-1}, T'_i]$ are also unacknowledged by $\Det_z$ at the same time $t$, from which we deduce that \eqref{eq:tcp:detzdef} for $\Det_{z'}$ contradicts the absence of acknowledgments of $\Det_z$ in $(T'_{i-1}, T'_i]$. We therefore have $n_z(\sigma) \geq n_{z'}(\sigma)$, yielding that 
    \begin{equation}\label{eq:tcp_decreasing_n}
        r \cdot n_r(\sigma) \leq \int_0^r n_z(\sigma) \, dz.
    \end{equation}
    
    We thus derive
    \begin{align*}
    \E\left[\Rand_r(\sigma)\right]
    &\leq \Opt(\sigma) + q_r\cdot\Det_r'(\sigma) + p(r)\cdot \int_{r}^1  n_z(\sigma) \, dz\\
    &= \Opt(\sigma) + q_r\cdot\left(n_r(\sigma) - \int_r^1 n_z(\sigma) \, dz\right) + p(r)\cdot \int_{r}^1  n_z(\sigma) \, dz\\
    &= \Opt(\sigma) + q_r\cdot n_r(\sigma) + \left(p(r)-q_r\right) \cdot \int_r^1 n_z(\sigma) \, dz\\
    & \overset{(a)}{=} \Opt(\sigma) + q_r\cdot n_r(\sigma) + \frac{q_r}{r} \cdot \int_r^1 n_z(\sigma) \, dz\\
    & \overset{(b)}{\leq} \Opt(\sigma) + \frac{q_r}{r} \cdot \int_0^r n_z(\sigma) \, dz + \frac{q_r}{r} \cdot \int_r^1 n_z(\sigma) \, dz \\
    &= \Opt(\sigma) + \frac{q_r}{r} \cdot \int_0^1 n_z(\sigma) \, dz\\
    & \overset{(c)}{\leq} \left(1 + \frac{q_r}{r}\right) \cdot \Opt(\sigma),
    \end{align*}
    where $(a)$ follows by $p(r) = (1 + 1/r) \cdot q_r$, $(b)$ by \eqref{eq:tcp_decreasing_n}, and $(c)$ by \eqref{eq:tcp_opt_bound}. 
    Notice that $1 + q_r/r = 1 + \frac{e^r}{(1+r) \cdot e - e^r}$.
\end{proof}

We can now prove \cref{thm:res_tcp}, restated below.

\tcpub*

\begin{proof}
    For any $z \in [r, 1]$, $\Det_z$ is $\gamma(z)$-dependable by \cref{lem:tcp_dependable} and therefore $\gamma(r)$-dependable. For any $r \in (0, 1]$, $\Rand_r$ is then supported only by $\gamma(r)$-dependable algorithms, and thus it is $\gamma(r)$-dependable. Moreover, by \cref{lem:tcp_ub}, the algorithm is $(1 + \frac{e^{r}}{(1+r) \cdot e - e^{r}})$-competitive.
\end{proof}
\section{Lower Bound for TCP Acknowledgment}
\label{sec:lb_tcp}

In this section, we show a nontrivial lower bound for the TCP acknowledgment problem. 
Recall that no algorithm is better than $2$-dependable~\citep{DoGoSc01} and no algorithm is better than \mbox{$e/(e-1)$}-competitive~\citep{Seiden00}. Here, we show a lower bound on the dependability-competitiveness tradeoff, which is stronger than these two bounds combined.

Similarly to \cref{thm:res_tcp}, we parametrize both the dependability and the resulting bound on the expected competitive ratio by a single real $r \in [(\sqrt5-1)/2, 1]$, setting
\begin{align*}
    c(r) & \define r^2 + r - 1 && \in [0,1], \\
    \gamma(r) & \define 1 + 1/r && \in [2, (3+\sqrt5)/2], \\
    \rho(r) & \define 1 + \frac{(2+r) \cdot r^{2+r}}{(1+r)^2 - r^{2+r}} .
\end{align*}
The range of $r$ is chosen so that $c(r) \in [0,1]$. We show that every $\gamma(r)$-dependable randomized algorithm for TCP acknowledgment has expected competitive ratio at least $\rho(r)$, which proves \cref{thm:res_tcp_lb}. See \cref{fig:tcp} for a plot of $\rho(r)$ versus $\gamma(r)$.

Throughout this section, we fix $r \in [(\sqrt5-1)/2, 1]$ and write $c \define c(r)$, $\gamma \define \gamma(r)$ and $\rho \define \rho(r)$. We extensively use the following relationships:
\[
    \gamma \geq 2, \qquad \frac{\gamma}{\gamma-1} = 1 + r \leq 2, \qquad 1 + c = r \cdot (1+r).
\]

The proof consists of two parts. First, in \cref{sec:latencies}, we study instances similar to those by \citet{Seiden00}, consisting of bursts of packets whose sizes grow quickly. We show that $\gamma$-dependability forces any deterministic algorithm to satisfy certain inequalities. In particular, such an algorithm has to acknowledge the first burst no earlier than at latency $c$; more complex inequalities hold also for the subsequent bursts. Second, in \cref{sec:theorem}, we apply the inequalities obtained this way for the first two bursts. We construct a probability distribution over instances with one or two bursts, on which every deterministic algorithm satisfying the bounds of \cref{sec:latencies} has expected competitive ratio at least~$\rho$. \cref{thm:res_tcp_lb} then follows by the application of the min-max principle (\cref{lem:minmax}).

\paragraph{Burst sequences.}

For a vector $\rvx$, we denote its length by $|\rvx|$ and its $i$-th component by $x_i$. We write $\perp$ for the vector of length $0$ and $\reals_{\geq 0}^* \define \bigcup_{k \geq 0} \reals_{\geq 0}^k$ for the set of all finite vectors with non-negative entries. For $\rvx \in \reals^k$ and $j \leq k$, we write $\rvx_{\leq j} \define \vct{x_1, \dots, x_j}$ for the prefix of $\rvx$ of length~$j$.

An instance $E(\rvs, \rvy)$, defined for a vector $\rvs \in \ints_{\geq 1}^n$ of \emph{burst sizes} and a vector $\rvy \in \reals_{\geq 0}^{n-1}$ of \emph{gaps}, consists of $n$ \emph{bursts}; burst~$i$ arrives at time $t_i$ and consists of $s_i$ packets. We define $t_1 = 0$ and $t_{i+1} = t_i + y_i / s_i$ for $i \geq 1$.

We say that an unacknowledged burst $i$ has \emph{latency} $z$ at time $t_i + z/s_i$, i.e., the latency of a burst is the total latency which its packets accrued so far. Note that burst $i+1$ arrives at latency $y_i$ of burst $i$. We also define the \emph{overhead} of burst~$i$ as $\theta_i \define (\sum_{j \leq i} s_j) / s_i \geq 1$. While burst~$i$ accrues latency $z$, the bursts $1, \dots, i$ together accrue latency $\theta_i \cdot z$. Note that $\theta_1 = 1$ and that $\theta_i$ is close to $1$ whenever burst $i$ is much larger than all the preceding ones together.

In all our constructions, $s_1 = 1$, i.e., the first burst is a single packet, and $E(\vct{1}, \perp)$ is the instance which consists of this packet only.

For $i \in \set{1, \dots, n}$, we write $\sigma_{\leq i} \define E(\rvs_{\leq i}, \rvy_{\leq i-1})$ for the instance consisting of the first $i$ bursts.

\paragraph{Prefixes and costs.}

Consider a deterministic algorithm \Det with finite dependability, and assume that an adversary creates an instance $E(\rvs, \rvy)$ \emph{adaptively against~\Det}: it chooses the size~$s_{i+1}$ of burst $i+1$ and places this burst at latency $y_i$ of burst $i$, and it fixes both values only after it observes the acknowledgment of burst $i$ by \Det. (Such acknowledgment must eventually occur, as the dependability of \Det is finite.) The process results in an instance $E(\rvs, \rvy)$ with $n \define |\rvs|$ bursts. We denote the latency at which \Det acknowledges burst $i \leq n$ by $x_i$, which yields the vector
\[
    \rvx(\Det, \rvs, \rvy) = \vct{x_1, \dots, x_n},
\]
for which $x_i < y_i$ for every $i < n$.

For brevity, we write $x_1(\Det)$ for the latency at which \Det acknowledges the single burst of $E(\vct{1},\perp)$, i.e., $\rvx(\Det, \vct{1}, \perp) = \vct{x_1(\Det)}$.

Fix $j < n$. As \Det is deterministic and online, when it is executed on instance $E(\rvs_{\leq j+1}, \rvy_{\leq j})$ (a~prefix of $E(\rvs, \rvy)$), it acknowledges all bursts \emph{at the same latencies} $x_1, \dots, x_{j+1}$ as on $E(\rvs, \rvy)$. 

For $i \in \set{0, \dots, n}$, we write $C_i \define \sum_{j \leq i} (1 + x_j)$ and when $i < n$, 
we let $Y_i \define \sum_{j \leq i} y_j$. In particular, $C_0 = 0$ and $Y_0 = 0$. Note that on the instance $E(\rvs_{\leq i}, \rvy_{\leq i-1})$ \Det pays $1$ for each of its first $i$ acknowledgments, and the packets of burst $j \leq i$ accrue latency $x_j$. Hence $C_i$ is the cost of \Det on such instance. 

\paragraph{Flat instances.}

The bounds which we present are the strongest when all overheads $\theta_i$ are close to $1$. The adversary can easily enforce this by making every burst much larger than all the previous ones together. 

For a real $\xi > 0$, we call an instance $E(\rvs, \rvy)$ \emph{$\xi$-flat} if $s_1 = 1$ and $\theta_i \leq 1 + \xi$ for every burst $i$. As $\theta_1 = 1$, the condition constrains only bursts $i \geq 2$, and the adversary may satisfy it simply by choosing $s_i \geq (\sum_{j<i} s_j) / \xi$. The gaps $y_i$ on flat instances may be arbitrary. 

For a deterministic algorithm \Det, an integer $n \geq 1$ and a real $\xi > 0$, we denote by $\gI_n(\Det, \xi)$ the set of all $\xi$-flat instances with $n$ bursts and gaps $y_1, \dots, y_{n-1} \leq 1$ that can be generated adaptively against \Det. We have $\gI_1(\Det, \xi) = \set{E(\vct{1}, \perp)}$ for every $\xi$. 

Note that an adaptively generated instance has to satisfy the condition $y_i > x_i$. Thus, if $\Det$ decides to acknowledge burst $i < n$ at latency $x_i \geq 1$, then the resulting instance does not belong to $\gI_n(\Det, \xi)$. In particular, if $x_1(\Det) \geq 1$, then $\gI_n(\Det, \xi)$ is empty for every $n \geq 2$.

\begin{definition}
\label{def:class_Deps}
For an integer $n \geq 1$ and a real $\eps > 0$, we let
\begin{equation*}
    \xi_n(\eps) \define \frac{r \cdot \eps}{(1+r)^2 \cdot n} ,
\end{equation*}
and we let $\gD^{(n)}_{\eps}$ be the set of all deterministic algorithms \Det which satisfy the following two properties.
\begin{itemize}
    \item \textup{(first burst)} $x_1(\Det) \geq c$, i.e., \Det acknowledges the single burst of $E(\vct{1}, \perp)$ no earlier than at latency $c$.
    \item \textup{(recursion)} For every instance in $\gI_n(\Det, \xi_n(\eps))$, it holds that
    \begin{equation*}
        C_i \geq (1+r) \cdot ( C_{i-1} + r - Y_{i-1} ) - \eps
        \qquad \text{for every $i \in \set{2, \dots, n}$.}
    \end{equation*}
\end{itemize}
\end{definition}

The recursion property lower-bounds the cost of \Det after each burst in terms of its cost after the previous one. Note that the first-burst property corresponds to the recursion property for $i = 1$, but without the error term: as $C_0 = Y_0 = 0$, we have $C_1 = 1 + x_1(\Det)$ and $(1+r) \cdot r = 1 + c$. The recursion property is vacuous whenever $\gI_n(\Det, \xi_n(\eps))$ is empty, e.g., if $x_1(\Det) \geq 1$. We will handle such deterministic algorithms in \cref{sec:theorem} by a separate straightforward argument. 

The lemma below states that classes $\gD^{(n)}_{\eps}$ contain all $\gamma$-dependable deterministic algorithms; its proof is deferred to \cref{sec:latencies}.

\begin{restatable}{lemma}{certificate}
    \label{lem:certificate}
    Consider a deterministic $\gamma$-dependable algorithm $\Det$, an integer $n \geq 1$, and a real $\eps > 0$. Then $\Det \in \gD^{(n)}_{\eps}$.
\end{restatable}

In \cref{sec:theorem}, we only need the properties of algorithms belonging to the class $\gD^{(2)}_{\eps}$ on sequences whose first gap is at most $1$. In this case, \cref{def:class_Deps} reduces to the following statement.

\begin{lemma}
    \label{lem:certificate2}
    Fix a real $\eps > 0$ and a deterministic algorithm $\Det \in \gD^{(2)}_{\eps}$, and let $x_1 \define x_1(\Det)$. Then $x_1 \geq c$ and, for every integer $s \geq 1/\xi_2(\eps)$ and every real $y \in (x_1, 1]$, it holds that
    \[
        \Det(E(\vct{1, s}, \vct{y})) \geq (1+r) \cdot (1 + x_1 + r - y) - \eps .
    \]
\end{lemma}

\begin{proof}
    The bound $x_1 \geq c$ is the first-burst property of \cref{def:class_Deps}. Next, fix $s$ and $y$ as in the statement. The instance $E(\vct{1, s}, \vct{y})$ is generated adaptively against \Det, as its second burst arrives at latency~$y > x_1$ of the first one. It is $\xi_2(\eps)$-flat, as $\theta_1 = 1$ and $\theta_2 = 1 + 1/s \leq 1 + \xi_2(\eps)$, and its only gap satisfies $y \leq 1$ by the lemma assumption. Thus, it belongs to $\gI_2(\Det, \xi_2(\eps))$, and therefore the recursion property of \cref{def:class_Deps} together with $C_1 = 1 + x_1$, $Y_1 = y$ and $C_2 = \Det(E(\vct{1, s}, \vct{y}))$ yields the lemma.
\end{proof}

\subsection{Structural Consequences of Dependability}
\label{sec:latencies}

In this section, we prove \cref{lem:certificate}. The idea behind our construction is the following: when \Det acknowledges too early, the adversary can append burst after burst, each of them arriving just after \Det acknowledges the previous one. We call such a chain of bursts a \emph{penalizing suffix}. This forces \Det to pay $1$ for an acknowledgment in every round, while $\Opt$ may acknowledge the whole sequence at the last burst only. To remain $\gamma$-dependable, \Det then has to make its acknowledgment latencies smaller and smaller, until this becomes impossible.

We make this argument applicable at an arbitrary moment of an arbitrary adaptively generated instance. More formally, we define the \emph{post-ack budget}, a quantity which is computed on the basis of the history. Whenever it drops below $1+r$, the adversary may append the penalizing suffix of bursts, which contradicts the $\gamma$-dependability of \Det. As this construction can be applied after each burst of a flat instance, we obtain one bound per burst, and these bounds are exactly the inequalities stated in \cref{def:class_Deps}.

Throughout this subsection, we fix deterministic algorithm \Det with finite dependability. All instances considered below are generated adaptively against \Det, and we let $\rvx(\Det, \rvs, \rvy) = \vct{x_1, \dots, x_n}$. Recall that $\Det(\sigma_{\leq i}) = C_i$.

\begin{lemma}
    \label{lem:opt_bound}
    Let $\sigma$ be an instance with $n$ bursts generated adaptively against \Det. For every $i \in \set{1, \dots, n}$, it holds that $\Opt(\sigma_{\leq i}) \leq 1 + \sum_{j < i} \theta_j \cdot y_j$ and $\Opt(\sigma_{\leq i}) \leq i$.
\end{lemma}

\begin{proof}
    The term $1 + \sum_{j<i} \theta_j \cdot y_j$ is the cost of the schedule that serves $\sigma_{\leq i}$ by a single acknowledgment at the time of burst $i$. Such schedule pays $1$ for the acknowledgment, and gathers latency $\theta_j \cdot y_j$ within the gap between bursts $j$ and $j+1$. Next, the term $i$ is the cost of the schedule that serves~$\sigma_{\leq i}$ by acknowledging at each of its $i$ bursts.
\end{proof}

\paragraph{Post-ack budgets.}

For an instance $\sigma$ with $n$ bursts and $i \in \set{1, \dots, n}$, we define the \emph{budget at burst $i$} as
\begin{equation}
\label{eq:phi_i}
    \Phi_i \define \gamma \cdot (1 + \sum_{j < i} \theta_j \cdot y_j) - C_{i-1} .
\end{equation}
By \cref{lem:opt_bound}, the first term is at least $\gamma \cdot \Opt(\sigma_{\leq i})$, i.e., at least the total cost which $\gamma$-dependability allows \Det to spend on $\sigma_{\leq i}$. $C_{i-1}$ is the amount which \Det has already paid before burst $i$ arrived, and thus $\Phi_i$ is the part which is not spent yet. Note that $\Phi_1 = \gamma$.

\begin{lemma} \label{lem:bounds}
    Let $\sigma$ be an instance with $n$ bursts generated adaptively against \Det. If \Det is $\gamma$-dependable, then for every $i \in \set{1, \dots, n}$, $x_i \leq \Phi_i - 1$ and $x_i \leq i/r$.
\end{lemma}

\begin{proof}
    Fix $i \in \set{1, \dots, n}$. As \Det is $\gamma$-dependable, $C_i = \Det(\sigma_{\leq i}) \leq \gamma \cdot \Opt(\sigma_{\leq i})$. Thus, by \cref{lem:opt_bound} and \eqref{eq:phi_i}, we have  $C_i \leq \gamma \cdot (1 + \sum_{j < i} \theta_j \cdot y_j) = \Phi_i + C_{i-1}$. As $C_i = C_{i-1} + 1 + x_i$, the first inequality of the lemma follows.
    
    By \cref{lem:opt_bound} again, $C_i \leq \gamma \cdot i$. As $C_{i-1} \geq i - 1$, we obtain $1 + x_i = C_i - C_{i-1} \leq \gamma \cdot i - (i - 1)$, and thus $x_i \leq (\gamma - 1) \cdot i = i/r$.
\end{proof}

By the budget definition \eqref{eq:phi_i}, for every instance with $n$ bursts and every $i < n$,
\begin{equation}
    \label{eq:step}
    \Phi_{i+1} = \Phi_i + \gamma \cdot \theta_i \cdot y_i - (1 + x_i) .
\end{equation}
Assume that \Det has just acknowledged burst $i$ at latency $x_i$ and that the adversary is about to issue burst $i+1$. The cheapest continuation for the adversary is to place this burst at latency $y_i$ which is arbitrarily close to $x_i$. By \eqref{eq:step}, the resulting budget $\Phi_{i+1}$ is then arbitrarily close to the amount denoted $\Phi_i^+$, which we call \emph{post-ack budget at burst $i$}:
\begin{align}
    \label{eq:phi_plus}
    \Phi^+_i
    & \define \Phi_i + \gamma \cdot \theta_i \cdot x_i - (1 + x_i) \nonumber \\
    & = \gamma \cdot (1 + \sum_{j < i} \theta_j \cdot y_j + \theta_i \cdot x_i) - C_i .
\end{align}
The second equality above follows by \eqref{eq:phi_i} and $C_i = C_{i-1} + 1 + x_i$. It is worth noting that $\Phi^+_i$ depends only on the history up to the acknowledgment of burst $i$ by \Det. In particular, it does not depend on the future choices of the adversary. By \cref{lem:bounds}, $\Phi^+_i \geq \gamma \cdot \theta_i \cdot x_i \geq 0$.

\paragraph{The penalizing suffix.}

Recall the definition \eqref{eq:phi_plus} of the post-ack budget. The lemma below states that the post-ack budget of a $\gamma$-dependable algorithm can never drop below $1+r$. 

The idea behind the lemma is that $1+r$ is the fixed point of the mapping $z \mapsto \gamma \cdot (z-1)$, which, up to lower-order terms, describes how the post-ack budget evolves from burst to burst once the adversary appends bursts greedily. We show that if the post-ack budget falls below $1+r$, then its distance to $1+r$ grows geometrically with each subsequent burst. Eventually, the budget drops below zero, which contradicts \cref{lem:bounds}.

We will use the following technical claim, whose proof is deferred to \cref{sec:appendix}.

\begin{restatable}{claim}{recurrence}
    \label{cla:recurrence}
    Fix reals $a \geq 2$ and $\beta \in (0,2]$, and let $\ell \define \lceil \log_a \frac{2}{\beta} \rceil$. If the sequence $z_0, z_1, \dots, z_\ell$ satisfies
    \begin{align*}
        z_0 & \leq \frac{a}{a-1} - 2 \cdot \beta, \\
        z_{i+1} & \leq a \cdot (z_i - 1) + \beta
            && \text{for every $i \in \set{0, \dots, \ell-1}$},
    \end{align*}
    then $z_\ell < 0$.
\end{restatable}

\begin{lemma}
    \label{lem:test}
    Let $\sigma$ be an instance with $n$ bursts generated adaptively against \Det. If \Det is $\gamma$-dependable, then $\Phi^+_k \geq 1+r$ for every $k \in \set{1, \dots, n}$.
\end{lemma}

\begin{proof}
    Assume towards contradiction that $\Phi^+_k < 1+r$ for some $k \leq n$. We show that a suitable adaptive extension of $\sigma_{\leq k}$ contradicts the $\gamma$-dependability of~\Det. To this end, we use the following constants.
    \begin{itemize}
        \item $\Delta \define (1+r) - \Phi^+_k$,
        \item $L \define \lceil \log_\gamma(8/\Delta) \rceil$ and $M \define k + L + 1$,
        \item $\delta \define \Delta / (8 \cdot \gamma \cdot \theta_k)$,
        \item $\eta \define \Delta / (8 \cdot \gamma \cdot (M/r + \delta))$.
    \end{itemize}
    As $0 < \Delta \leq 1 + r \leq 2$, all constants are positive and $L \geq 1$. Note that they depend on $\sigma_{\leq k}$ and on the behavior of \Det on it only.

    The extension consists of bursts $k+1, \dots, M$. For $m \in \set{k+1, \dots, M}$, once \Det acknowledges burst~$m-1$ at latency $x_{m-1}$, the adversary issues burst $m$ at latency $y_{m-1} \define x_{m-1} + \delta$ of burst $m-1$, and it chooses its size $s_m \define \lceil (\sum_{j < m} s_j) / \eta \rceil$. Thus, $\theta_m \leq 1 + \eta$ for every $m \in \set{k+1, \dots, M}$, while the overheads $\theta_1, \dots, \theta_k$ inherited from $\sigma_{\leq k}$ can be arbitrary.

    We denote the resulting instance by $\sigma'$. As $\sigma'$ and $\sigma$ share the prefix $\sigma_{\leq k}$, \Det acknowledges the first $k$ bursts of $\sigma'$ at the same latencies $x_1, \dots, x_k$, and $\Phi_k$ and $\Phi^+_k$ are the same for both instances. We first bound $\Phi_{k+1}$ as 
    \begin{align*}
        \Phi_{k+1}
        & = \Phi_k + \gamma \cdot \theta_k \cdot (x_k + \delta) - (1 + x_k) 
            && \text{(by \eqref{eq:step} and $y_k = x_k + \delta$)} \\
        & = \Phi^+_k + \gamma \cdot \theta_k \cdot \delta 
            && \text{(by the definition of $\Phi^+_k$)} \\ 
        & = (1+r) - \Delta + \Delta/8 
            && \text{(by the definitions of $\Delta$ and $\delta$)} \\
        & \leq (1+r) - \Delta / 2.
    \end{align*}
    Next, for $m \in \set{k+1, \dots, M-1}$, we bound $\Phi_{m+1}$ in terms of $\Phi_m$. By \cref{lem:bounds}, $x_m \leq m/r \leq M/r$, and thus 
    \begin{equation}
        \label{eq:eta_helper}
        \eta \cdot \gamma \cdot (x_m + \delta) \leq 
        \eta \cdot \gamma \cdot (M/r + \delta) =
        \Delta / 8.
    \end{equation}
    We obtain
    \begin{align*}
        \Phi_{m+1}
        & = \Phi_m + \gamma \cdot \theta_m \cdot (x_m + \delta) - (1 + x_m)
            && \text{(by \eqref{eq:step} and $y_m = x_m + \delta$)} \\
        & \leq \Phi_m + \gamma \cdot (1+\eta) \cdot (x_m + \delta) - (1 + x_m)
            && \text{(as $\theta_m \leq 1 + \eta$)} \\
        & = \Phi_m + (\gamma-1) \cdot x_m - 1 + \gamma \cdot \delta + \gamma \cdot \eta \cdot (x_m + \delta) \\
        & \leq \Phi_m + (\gamma-1) \cdot x_m - 1 + \Delta/8 + \Delta/8
            && \text{(as $\theta_k \geq 1$ and by \eqref{eq:eta_helper})} \\
        & \leq \gamma \cdot (\Phi_m - 1) + \Delta/4 ,
    \end{align*}
    where the last inequality follows by \cref{lem:bounds}, i.e., $x_m \leq \Phi_m - 1$.

    Now we apply \cref{cla:recurrence} to the sequence $z_i \define \Phi_{k+1+i}$ for $i \in \set{0, \dots, L}$, with $a = \gamma$, $\beta = \Delta/4$ and $\ell = \lceil \log_a(2/\beta) \rceil = L$; the two bounds above are its assumptions. The claim gives $\Phi_M - 1 = z_L - 1 < -1 < x_M$, which contradicts \cref{lem:bounds} for burst $M$.
\end{proof}

\paragraph{Constraints from all bursts.}

By applying \cref{lem:test} after every burst of a flat instance, we obtain one inequality per burst, which proves \cref{lem:certificate}, restated below.

\certificate*

\begin{proof}
    Consider an arbitrary instance $\sigma$ generated adaptively against \Det. By \cref{lem:test}, for every burst $i$ of $\sigma$, it holds that $\Phi^+_i \geq 1+r$, and thus \eqref{eq:phi_plus} implies
    \begin{equation}
        \label{eq:phi_plus2}
        C_i \leq \gamma \cdot (1 + \sum_{j < i} \theta_j \cdot y_j + \theta_i \cdot x_i) - (1+r) .
    \end{equation}

    We start with showing the first-burst property. On the instance $E(\vct{1}, \perp)$, we have $C_1 = 1 + x_1(\Det)$ and $\theta_1 = 1$, so \eqref{eq:phi_plus2} gives $C_1 \leq \gamma \cdot C_1 - (1+r)$. This gives us $C_1 \geq (1+r) / (\gamma - 1) = (1+r) \cdot r = 1 + c$, i.e., $x_1(\Det) \geq c$.

    It remains to show the recursion property. Let $\xi \define \xi_n(\eps)$, cf.~\cref{def:class_Deps}. Consider an instance $E(\rvs, \rvy) \in \gI_n(\Det, \xi)$ with $\rvx(\Det, \rvs, \rvy) = \vct{x_1, \dots, x_n}$, and fix $i \in \set{2, \dots, n}$. As all gaps are at most $1$, we have $Y_{i-1} \leq i - 1$, and by \cref{lem:bounds}, $x_i \leq i/r$. Thus, $Y_{i-1} + x_i \leq n \cdot (1 + 1/r)$, and hence
    \begin{equation}
        \label{eq:xi_error}
        \gamma \cdot \xi \cdot (Y_{i-1} + x_i) \leq \gamma \cdot \xi \cdot n \cdot (1 + 1/r) = \eps / r .
    \end{equation}
    Now, using \eqref{eq:phi_plus2} we have
    \begin{align*}
        C_i
        & \leq \gamma \cdot (1 + Y_{i-1} + x_i) + \gamma \cdot \xi \cdot (Y_{i-1} + x_i) - (1+r)
            && \text{(as $\theta_j \leq 1 + \xi$)} \\
        & \leq \gamma \cdot (1 + Y_{i-1} + x_i) - (1+r) + \eps/r
            && \text{(by \eqref{eq:xi_error})} \\
        & = \gamma \cdot (Y_{i-1} + C_i - C_{i-1}) - (1+r) + \eps/r .
            && \text{(as $1 + x_i = C_i - C_{i-1}$)}
    \end{align*}
    By reorganizing the terms, we get $(\gamma - 1) \cdot C_i \geq \gamma \cdot (C_{i-1} - Y_{i-1}) + (1+r) - \eps/r$. The lemma follows by multiplying both sides by $r = 1/(\gamma-1)$ and using $r \cdot \gamma = 1+r$.
\end{proof}

\subsection{Randomized Lower Bound}
\label{sec:theorem}

The argument below uses properties of deterministic algorithm $\Det \in \gD^{(2)}_\eps$ given by \cref{lem:certificate2}: the lower bound on $x_1(\Det)$ and the lower bound on its cost on a two-burst instance.

We note that the first bound alone already yields a nontrivial lower bound. Indeed, consider a $\gamma$-dependable randomized algorithm \Rand. Every $\Det \in \supp(\Rand)$ is $\gamma$-dependable, so $x_1(\Det) \geq c$ by \cref{lem:certificate}. On the instance $E(\vct{1}, \perp)$, $\Opt$ pays $1$ and \Det pays $1 + x_1(\Det) \geq 1 + c$. Thus, the expected competitive ratio of \Rand is at least $1 + c = r \cdot (1+r)$. In what follows, we construct a distribution over instances on which the two bounds together imply \cref{thm:res_tcp_lb}, restated below for convenience.

\tcplb*
 
\paragraph{Probability distribution.}

Probability distribution $\pi_\eps$ defined below consists of two parts: an atom on the single-packet instance and a density over two-burst instances, in which the second burst is large enough to make the instance flat, and the gap $y$ between the two bursts ranges over $[c, r]$.

For $y \in [c, r]$, we write $w(y) \define (1+r) - (1+y)/(1+r)$. The function $w$ is decreasing, $w(c) = 1$ as $1 + c = r \cdot (1+r)$, and $w(r) = r$, i.e., $w$ maps $[c, r]$ onto $[r, 1]$. Let
\begin{align*}
    q 
    & \define \frac{(2+r) \cdot r^{1+r}}{(1+r)^2 - r^{2+r}} \text{ and} \\
    p(y) 
    & \define q \cdot \frac{(1 + r - w(y)) \cdot w(y)^r}{r^{1+r}} 
    = q \cdot \frac{(1+y) \cdot w(y)^r}{(1+r) \cdot r^{1+r}} 
        \quad \text{for $y \in [c, r]$} .
\end{align*}
Note that $\rho = 1 + q \cdot r$. Fix a real $\eps > 0$ and let $s \define \lceil 1/\xi_2(\eps) \rceil$, cf.~\cref{def:class_Deps}. We define $\pi_\eps$ as the following probability distribution:
\begin{itemize}
    \item with probability $q$, the instance is $E_1 = E(\vct{1}, \perp)$;
    \item with density $p(y)$ on $y \in [c, r]$, the instance is $E_2 = E(\vct{1, s}, \vct{y})$.
\end{itemize}

Note that $\pi_\eps$ is indeed a probability distribution.

\begin{claim}
    \label{cla:normalization}
    It holds that $q + \int_c^r p(y) \, dy = 1$.
\end{claim}

\begin{proof}
    We substitute $w = w(y)$, so that $dy = -(1+r) \, dw$, and $w$ runs from $1$ down to $r$ as $y$ runs from $c$ to $r$. Thus,
    \begin{align*}
        \int_c^r p(y) \, dy
        & = \frac{q \cdot (1+r)}{r^{1+r}} \cdot \int_r^1 (1 + r - w) \cdot w^r \, dw \\
        & = \frac{q \cdot (1+r)}{r^{1+r}} \cdot \left[ w^{1+r} - \frac{w^{2+r}}{2+r} \right]_r^1 \\
        & = \frac{q \cdot (1+r)}{r^{1+r}} \cdot \frac{(1+r) - 2 \cdot r^{1+r}}{2+r} \\
        & = q \cdot \left( \frac{(1+r)^2}{(2+r) \cdot r^{1+r}} - \frac{2 \cdot (1+r)}{2+r} \right) \\
        & = q \cdot \left( \frac{1}{q} + \frac{r}{2+r} - \frac{2 \cdot (1+r)}{2+r} \right)
            && \text{(by the definition of $q$)} \\
        & = q \cdot (1/q - 1) = 1 - q . && \qedhere
    \end{align*}
\end{proof}

Below, we show that, up to an additive $\eps$, the expected ratio on $\pi_\eps$ of any algorithm from $\gD^{(2)}_\eps$ which acknowledges the first burst at latency $t \in [c, r]$ is at least $F(t)$, where
\[
    F(t) \define q \cdot (1 + t) + \int_c^t p(y) \, dy + (1+r) \cdot \int_t^r \frac{1 + t + r - y}{1+y} \cdot p(y) \, dy .
\]
\begin{claim}
    \label{cla:constant}
    It holds that $F(t) = \rho$ for every $t \in [c, r]$.
\end{claim}

\begin{proof}
    By substituting $u = w(y)$, we obtain, for every $t \in [c, r]$,
    \[
        \int_t^r w(y)^r \, dy = (1+r) \cdot \int_r^{w(t)} u^r \, du = w(t)^{1+r} - r^{1+r} .
    \]
    Next, by the definition of $p$, we have
    \[
        F(t) = q \cdot (1 + t) + \int_c^t p(y) \, dy + \frac{q}{r^{1+r}} \cdot \int_t^r (1 + t + r - y) \cdot w(y)^r \, dy .
    \]
    We show that $F' = 0$ on $[c, r]$. Fix $t \in [c, r]$. At $y = t$, the integrand of the last integral equals $(1+r) \cdot w(t)^r$, and its derivative with respect to $t$ is $w(y)^r$. Thus,
    \begin{align*}
        F'(t) 
        & = q + p(t) - \frac{q \cdot (1+r) \cdot w(t)^r}{r^{1+r}} + \frac{q}{r^{1+r}} \cdot \int_t^r w(y)^r \, dy \\
        & = q + p(t) - \frac{q \cdot (1+r) \cdot w(t)^r}{r^{1+r}} + \frac{q \cdot w(t)^{1+r}}{r^{1+r}} - q \\
        & = p(t) - q \cdot \frac{(1 + r - w(t)) \cdot w(t)^r}{r^{1+r}} \\
        & = 0 .
            && \text{(by the definition of $p$)}
    \end{align*}
    Thus, $F$ is constant on $[c, r]$, and
    \begin{align*}
        F(t) = F(r) 
        & = q \cdot (1 + r) + \int_c^r p(y) \, dy \\
        & = q \cdot (1+r) + 1 - q 
            && \text{(by \cref{cla:normalization})} \\
        & = \rho 
            && \text{(by the definition of $\rho$).}
        \qedhere
    \end{align*}
\end{proof}

\begin{lemma}
    \label{lem:bad_distribution}
    For every deterministic algorithm $\Det \in \gD^{(2)}_{\eps}$, it holds that $\E_{\sigma \sim \pi_\eps}[\ratio(\Det, \sigma)] \geq \rho - \eps$.
\end{lemma}

\begin{proof}
    Consider a deterministic algorithm $\Det \in \gD^{(2)}_{\eps}$ and let $x_1 \define x_1(\Det)$; by \cref{lem:certificate2}, $x_1 \geq c$. We may assume that $x_1$ is finite, as otherwise $\ratio(\Det, E_1) = \infty$ and the lemma holds trivially. We bound the ratio of \Det in the following three cases.
    \begin{itemize}
        \item On $E_1$: \Det pays $1 + x_1$ and $\Opt$ pays $1$, so $\ratio(\Det, E_1) = 1 + x_1$.
        \item On $E_2$ with $y \leq x_1$: the second burst arrives before \Det acknowledges the first one, so \cref{lem:certificate2} does not apply, and we use the trivial bound $\ratio(\Det, E_2) \geq 1$. This case covers also all algorithms with $x_1 \geq 1$.
        \item On $E_2$ with $y > x_1$: as $y \leq r \leq 1$, \cref{lem:certificate2} yields $\Det(E_2) \geq (1+r) \cdot (1 + x_1 + r - y) - \eps$, and $\Opt$ may acknowledge once at the second burst, paying $1 + y$. Thus, as $1 + y \geq 1$,
        \[
            \ratio(\Det, E_2)
            \geq \frac{(1+r) \cdot (1 + x_1 + r - y) - \eps}{1+y}
            \geq \frac{(1+r) \cdot (1 + x_1 + r - y)}{1+y} - \eps .
        \]
    \end{itemize}

    First, consider the case $x_1 > r$. As $y$ is chosen from $[c, r]$, we have $y \leq r < x_1$, so the second case above applies to every two-burst instance, and $\E_{\sigma \sim \pi_\eps}[\ratio(\Det,\sigma)] \geq q \cdot (1 + x_1) + (1 - q) = 1 + q \cdot x_1 > 1 + q \cdot r = \rho$.

    Second, consider the case $x_1 \in [c, r]$. The three cases above yield
    \begin{align*}
        \E_{\sigma \sim \pi_\eps}[\ratio(\Det,\sigma)]
        & \geq F(x_1) - \eps \cdot \int_{x_1}^r p(y) \, dy \\
        & \geq F(x_1) - \eps \\
        & = \rho - \eps. 
            && \text{(by \cref{cla:constant})}
        \qedhere
    \end{align*}
\end{proof}

\begin{proof}[Proof of \cref{thm:res_tcp_lb}]
    Fix a real $\eps > 0$. By \cref{lem:certificate}, every $\gamma$-dependable deterministic algorithm belongs to $\gD^{(2)}_{\eps}$, and by \cref{lem:bad_distribution}, each of them has expected ratio at least $\rho - \eps$ on $\pi_\eps$. The min-max principle (\cref{lem:minmax}) with $\gD = \gD^{(2)}_{\eps}$ gives the same bound for every $\gamma$-dependable randomized algorithm. As $\eps > 0$ can be arbitrarily small, the theorem follows.
\end{proof}

\subsection{Proofs of the Technical Claims}
\label{sec:appendix}

\recurrence*

\begin{proof}
    For every $i \in \set{0, \dots, \ell}$, we define the \emph{gap} $g_i \define \frac{a - \beta}{a-1} - z_i$. Using $a \geq 2$, we have
    \begin{align*}
        g_0
        & = \frac{a - \beta}{a-1} - z_0
        \geq \frac{a - \beta}{a-1} - \frac{a}{a-1} + 2 \cdot \beta
        = - \frac{\beta}{a-1} + 2 \cdot \beta \geq -\beta + 2 \cdot \beta = \beta,
    \intertext{and for every $i \in \set{0, \dots, \ell-1}$,}
        g_{i+1}
        & = \frac{a - \beta}{a-1} - z_{i+1}
        \geq \frac{a - \beta}{a-1} - a \cdot (z_i - 1) - \beta \\
        & = \frac{a - \beta}{a-1} + a - \beta - a \cdot z_i
        = a \cdot \left( \frac{a - \beta}{a-1} - z_i \right)
        = a \cdot g_i.
    \end{align*}
    Thus, $g_i \geq a^i \cdot \beta$ for every $i \in \set{0, \dots, \ell}$. In particular, $g_\ell \geq a^\ell \cdot \beta \geq (2/\beta) \cdot \beta = 2$. Furthermore, $a/(a-1) \leq 2$ as $a \geq 2$, and thus
    \[
        z_\ell
        = \frac{a}{a-1} - \frac{\beta}{a-1} - g_\ell
        \leq 2 - \frac{\beta}{a-1} - 2
        < 0 . \qedhere
    \]
\end{proof}

\end{document}